\documentclass[
final
 , nomarks
]{dmtcs-episciences}
\usepackage{latexsym,amssymb,amsmath,enumerate,geometry,float,tikz,setspace,verbatim,amsthm}
\usepackage{graphicx}
\usepackage[modulo]{lineno}
\newtheorem{claim}{Claim}

\usepackage[numbers]{natbib}
\usepackage{comment}
\usepackage{amsmath,amssymb}
\usepackage{mathtools}
\usepackage{thmtools,thm-restate}
\declaretheorem[name=Theorem,numberwithin=section]{mytheorem}

\declaretheorem[name=Corollary,sibling=mytheorem]{mycorollary}
\usepackage{enumitem}
\usepackage{tabularx,environ}
\usepackage[normalem]{ulem}
\usepackage[labelformat=simple]{subcaption}

\usepackage{tikz}
\usetikzlibrary{matrix,fit,positioning,calc}
\usepackage{float}
\usepackage{soul} 

\usepackage{complexity}

\newclass{\Hard}{hard}
\newclass{\Hness}{hardness}

\newclass{\para}{para}
\newclass{\Complete}{complete}
\newclass{\Cness}{completeness}

\definecolor{blue-violet}{rgb}{0.54, 0.17, 0.89}
\definecolor{mangotango}{rgb}{1.0, 0.51, 0.26}

\makeatletter
\newcommand{\problemtitle}[1]{\gdef\@problemtitle{#1}}
\newcommand{\probleminput}[1]{\gdef\@probleminput{#1}}
\newcommand{\problemquestion}[1]{\gdef\@problemquestion{#1}}
\NewEnviron{prb}{
  \problemtitle{}\probleminput{}\problemquestion{}
  \BODY
  \par\addvspace{.5\baselineskip}
  \noindent
  \begin{tabularx}{\textwidth}{@{\hspace{\parindent}} l X c}
    \multicolumn{2}{@{\hspace{\parindent}}l}{\@problemtitle} \\
    \textbf{Input:} & \@probleminput \\
    \textbf{Question:} & \@problemquestion
  \end{tabularx}
  \par\addvspace{.5\baselineskip}
}
\makeatother

\author{M\'{a}rcia R. Cappelle
  \and Erika Coelho
  \and Les R. Foulds
  \and Humberto J. Longo}
\title[Open-independent, open-locating-dominating sets:  structural aspects of some classes of graphs]{Open-independent, open-locating-dominating sets:  structural aspects of some classes of graphs}
\affiliation{
  Instituto de Inform\'{a}tica, Universidade Federal de Goi\'{a}s, Goi\^{a}nia, Brazil
}
\received{2021-09-02}
\revised{2021-12-22}
\accepted{2022-01-05}
\begin{document}
\publicationdetails{24}{2022}{1}{5}{8440}
\maketitle

\begin{abstract}
Let $G=(V(G),E(G))$ be a finite simple undirected graph with vertex set $V(G)$, edge set $E(G)$ and vertex subset $S\subseteq V(G)$. $S$ is termed \emph{open-dominating} if every vertex of $G$ has at least one neighbor in $S$, and \emph{open-independent, open-locating-dominating} (an $OLD_{oind}$-set for short) if no two vertices in $G$  have the same set of neighbors in $S$, and each vertex in $S$ is open-dominated exactly once by $S$. The problem of deciding whether or not $G$ has an $OLD_{oind}$-set has important applications that have been reported elsewhere. As the problem is known to be $\mathcal{NP}$-complete, it appears  to be notoriously difficult as we show that its complexity remains the same even for just planar bipartite graphs of maximum degree five and girth six, and also for planar subcubic graphs of girth nine. Also, we present characterizations of both $P_4$-tidy graphs and the complementary prisms of cographs that have an $OLD_{oind}$-set.
\end{abstract}

\section{Introduction}
\label{intro}

Consider the situation where a graph $G$ models a facility or a multiprocessor network with limited-range detection devices (sensing for example, movement, heat or size) that are placed at chosen vertices of $G$. The purpose of these devices is to detect and precisely identify the location of an intruder such as a thief, saboteur, vandal, fire or faulty processor that may suddenly be present at any vertex.  

As it is costly to install and maintain such devices it is logical to determine the locations of the minimum number of devices that can, between them, precisely determine an intruder at any location.  This challenge is often called a \emph{location-detection} or an \emph{identification} problem and has been well-studied \cite{honkala2004locating,laifenfeld2009joint,ray2004robust}.  This objective is adopted throughout the present article. Sometimes such a device can determine if an intruder is in its neighborhood but cannot detect if the intruder is at its own location. In this case, it is required to find a so-called, \emph{open-locating-dominating} vertex subset $S$ (an $OLD$-set for short), which is a dominating set of $G$, such that every vertex in $G$ has at least one neighbor in $S$, and no two vertices in $G$ have the same set of neighbors in $S$. When a device may be prevented from detecting an intruder at its own location, it is necessary to install another device in its neighborhood. A natural way to analyze such situations is to make use of open neighborhood sets which may have useful additional properties, such as being open-independent, dominating, open-dominating or open-locating-dominating. A set $S$ is \emph{open-independent} if every member of $S$ has at most one neighbor in $S$. The other terms and those in the next paragraph are made more precise later in this section.

An alternative case arises when a device can determine if the intruder is either at its own location or is in the neighborhood of its location, but which actual location cannot be detected, and furthermore, each detector cannot be located in the range of any other due to possible signal interference. 
Such situations can potentially be analyzed via independent, locating-dominating sets \cite{slater2018}. 

Finding an $OLD$-set in a given graph, if it exists, is similar to the well-studied \emph{identifying code problem} \cite{Karpovsky-Chakrabarty-Levitin-1998}. An \emph{identifying code} $C\subseteq V(G)$ is a dominating set where for all $u,v \in V(G)$, $u \neq v$, one has $N[u]\cap C \neq N[v] \cap C$. The minimum cardinality of an identifying code of a graph $G$ is denoted by $IC(G)$. Although  $OLD$-sets and identifying codes are similar notions, the parameters $OLD(G)$ and $IC(G)$ are incomparable.
The concept of an open-locating-dominating set was first considered by Seo and Slater \cite{Seo-Slater-2010,seo2011}. The authors showed that to decide if a graph $G$ has such a set is an $\mathcal{NP}$-complete decision problem and they provided some useful results for $OLD$-sets in trees and grid graphs. Foucaud et al. \cite{foucaud2017-i} presented a linear time algorithm that can be used to construct an open-locating dominating set of minimum size for a given cograph, based on parsing its cotree structure. 
Kincaid et al. \cite{kincaid2015} established the optimal density of an OLD-set for infinite triangular grids. Savic et al. \cite{savic2018open} presented results on OLD-sets for some convex polytopes. The reader is referred to \cite{lobstein2012watching} for an on-line bibliography on this topic and related notions.
In this paper we consider the following more restrictive problem:

  \begin{flushleft}
    \fbox{
      \begin{minipage}{0.95\linewidth}
        \noindent {\textsc{OLD-OIND} (existence of an open-independent, open locating dominating set)}\\
        \textbf{Instance:} {A graph $G$.}\\
        \textbf{Question:} {Does $G$ have an open-independent, open locating dominating set?}
      \end{minipage}}
    \medskip
  \end{flushleft} 
  
To the best of our knowledge, Seo and Slater \cite{Seo-Slater-2017} were the first to study open-independent, open-locating-dominating sets ($OLD_{oind}$-sets for short). They presented some results on $OLD_{oind}$-sets in paths, trees and infinite grid graphs, and characterized $OLD_{oind}$-sets in graphs with girth at least five. The authors also demonstrated that \textsc{OLD-OIND} is $\mathcal{NP}$-complete.
This complexity result was extended for complementary prisms by Cappelle et al. \cite{cappelle2020} who presented various properties and bounds on the sizes of minimal $OLD_{oind}$-sets in complementary prism graphs and showed that, if the girth of $G$ is at least four, the $OLD_{oind}$-set of its complementary prism, if it exists, can be found in polynomial time.
 
In this paper we analyze the existence of $OLD_{oind}$-sets in $P_4$-tidy graphs and in cographs, a subclass of the $P_4$-tidy class. The class of $P_4$-tidy graphs contains several other graph families having relatively few $P_4$’s, such as the $P_4$-sparse, $P_4$-lite, $P_4$-extendable and $P_4$-reducible graph families. The $P_4$-tidy graph class generalizes all of the just-mentioned graph families. It is well known that the $P_4$-tidy graph class is self-complementary and hereditary \cite{giakoumakis1997}.
 
Haynes et al. \cite{Haynes-Henning-Slater-Merwe-2007} investigated several graph theoretic properties of complementary prisms, such as independence, distance and domination. For further study on domination parameters in complementary prisms, see \cite{Gongora-Haynes-Jum-2013,Desormeaux-Haynes-2011,Desormeaux-Haynes-Vaughan-2013,Haynes-Henning-Merwe-2009,Haynes-Holmes-Koessler-2010} and for certain other parameters see \cite{bendali2019some,castonguay2019geodetic,meierling2015cycles,zatesko2019chromatic}. Cappelle et al. \cite{Cappelle-Penso-Rautenbach-2014} described a polynomial-time  recognition algorithm for complementary prisms. Although complementary prisms are a class of apparently well-behaved graphs, many $\mathcal{NP}$-complete problems for general graphs remain $\mathcal{NP}$-complete for this class, for example, finding an independent or a dominating set, or establishing $P_3$-convexity \cite{Duarte-Penso-Rautenbach-Santos-2017}.

\smallskip
\noindent \textit{Our contributions.}
It appears that \textsc{OLD-OIND} is notoriously difficult as we show that it remains $\mathcal{NP}$-complete even for just planar bipartite graphs of maximum degree five and girth six, and also for planar subcubic graphs of girth nine. However, we study some graph classes for which the problem can be solved in polynomial time and also present characterizations of both $P_4$-tidy graphs and complementary prisms of cographs that have an $OLD_{oind}$-set. 

\smallskip
\noindent \textit{Notation and terminology.} Throughout this paper $G = (V(G), E(G))$ is assumed to be a nontrivial finite simple undirected connected graph with vertex set $V(G)$ and edge set $E(G)$. A subgraph of $G$ with $n$ vertices that is a path (a cycle) is termed an $n$-path (an $n$-cycle), and is denoted by $P_n$ ($C_n$). A \textit{subcubic} graph is a graph in which each vertex has degree at most three, \textit{i.e.} no vertex is incident with more than three edges. The \emph{open neighborhood} of a vertex $v\in V(G)$ is denoted by $N_G(v) = \{u \in V(G) \mid uv \in E(G)\}$, and its \emph{closed neighborhood} is denoted by $N_G[v] = N_G(v) \cup \{v\}$. Let $S\subseteq V(G)$, $S \neq \emptyset$. The subgraph of $G$ induced by $S$ is denoted by $G[S]$. $S$ is termed \emph{dominating} if every vertex not in $S$ is adjacent to at least one member of $S$, \textit{i.e.}, $N_G[S] = \cup_{v\in S}N[v] = V(G)$. $S$ is \emph{locating-dominating} if it is dominating and no two distinct vertices of $V(G)\setminus{S}$ have the same set of neighbors in $S$, \textit{i.e.}, $\forall u,v \in V(G)\setminus S$, $u \neq v$, we have that $N(u)\cap S \neq N(v) \cap S$. Clearly, $V(G)$ is a locating-dominating set for any graph $G$. 
 
$S$ is \emph{open-dominating} (or \emph{total-dominating}) if every vertex $v \in V(G)$ has a neighbor in $S$, \textit{i.e.}, $\cup_{u\in S}N(u) = V(G)$. In this case $v$ is said to be \emph{open-dominated} by $S$. $S$ is \emph{independent} if no two vertices in $S$ are adjacent and is  \emph{open-independent} (an $OIND$-set for short) if every member of $S$ is open-dominated by $S$ at most once, \textit{i.e.}, $\forall\ v \in S,\ |N(v) \cap S| \leq 1$. Obviously, an open-dominating set cannot be independent. $S$ is \emph{open-locating-dominating} (an $OLD$-set for short) if it is open-dominating and no two distinct members of $V(G)$ have the same open neighborhood in $S$, \textit{i.e.}, for all $u,v \in V(G)$, $u \neq v$, one has $N(u)\cap S \neq N(v) \cap S$. In this case, $u$ and $v$ are said to be \emph{distinguished} by $S$. $G$ has an $OLD$-set if no two of its vertices have the same open neighborhood. The minimum cardinality of an $OLD$-set is denoted by $OLD(G)$, the open-locating-dominating number of $G$. If $G$ has an $OLD$-set $S$, then $S$ is called an $OLD(G)$-set if $|S| = OLD(G)$. We note that if $S$ exists, then the fact that every leaf and its support of $G$ must be in $S$ is helpful in the proofs of some of the theorems stated later in the present paper.

Let $G_1$ and $G_2$ be two graphs such that $V(G_1) \cap V(G_2) = \emptyset$. We denote the disjoint union (resp. join) of $G_1$ and $G_2$ by $G_1\oplus G_2$ (resp. $G_1\bowtie G_2$), and the complement graph of a graph $G$ by $\overline{G}$. 
An \emph{anticomponent} of a graph $G$ is the subgraph of $G$ induced by the vertex set of a connected component of $\overline{G}$. More precisely, an induced subgraph $H$ of $G$ is an anticomponent if $\overline{H}$ is a connected component of $G$. Notice that if $G_1,G_2,\ldots,G_k$ are the anticomponents of $G$, then $G = G_1 \bowtie G_2 \bowtie\cdots \bowtie G_k$.
A graph $G$ is termed \emph{co-connected} if $\overline{G}$ is connected. If $X\subseteq V(G)$, the subgraph of $G$ obtained by removing both the vertices in $X$ and the edges incident with them in $G$ is denoted by $G\setminus X$.
For a graph $H$ and a natural number $p$, the graph obtained by the disjoint union of $p$ copies of $H$ is denoted by $pH$. For integer $k\geq1$, the set $\{1,\ldots,k\}$ is denoted by $[k]$.

If an open-independent, open-locating-dominating set (an $OLD_{oind}$-set for short) exists in a given graph $G$, it is often of interest to establish a set of minimum size  among such sets in $G$, which is denoted by $OLD_{oind}(G)$. If $S$ is an $OLD_{oind}$-set for $G$, each component of $G[S]$ is isomorphic to $K_2$ (the complete graph on two vertices). 
See, for example, the graphs in Figures \ref{fig:ex1a} and \ref{fig:ex1b}, where an $OLD_{oind}$-set of each graph is represented by the black vertices. 

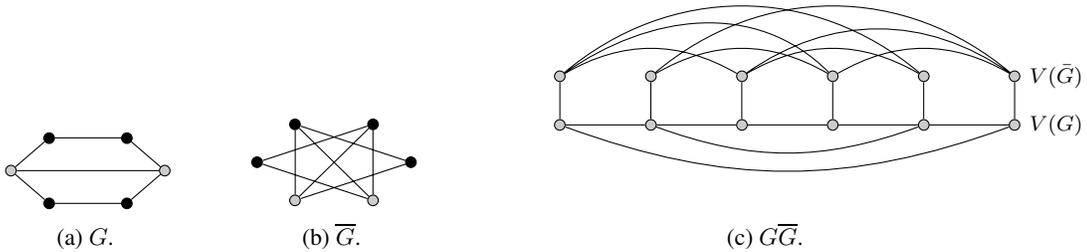
\begin{figure}[!ht]
\centering
 \begin{subfigure}[b]{.17\textwidth}
 \centering
  \begin{tikzpicture}[
  graph/.style={matrix of math nodes, ampersand replacement=\&, column sep=10pt, row sep=8pt, nodes={circle,inner sep=0pt,minimum size=4pt}},
  vSet/.style args = {(#1,#2)}{row #1 column #2/.style={nodes={draw,thin,fill=black!20!white}}},
  oldSet/.style args = {(#1,#2)}{row #1 column #2/.style={nodes={draw,thin,fill=black}}},
  oldSet/.list={(1,2),(1,3),(3,2),(3,3)},
  vSet/.list={(2,1),(2,4)},
]
\path[use as bounding box] (-1.70,-0.5) rectangle (0.7,0.5);
 \matrix [graph] (M)
 {
   {}\& {}\&[15pt] {}\& {}\& {}\&\\
   {}\& {}\& {}\& {}\& {}\&\\
   {}\& {}\& {}\& {}\& {}\&\\
 };
 \draw [thin]
   (M-1-2) -- (M-1-3)
   (M-1-3) -- (M-2-4)
   (M-2-1) -- (M-1-2)
   (M-2-1) -- (M-2-4)
   (M-2-1) -- (M-3-2)
   (M-3-2) -- (M-3-3)
   (M-3-3) -- (M-2-4);
\end{tikzpicture}
  \caption{$G$.}
  \label{fig:ex1a}
 \end{subfigure}\qquad
 \begin{subfigure}[b]{.17\textwidth}
 \centering
 \begin{tikzpicture}[
  graph/.style={matrix of math nodes, ampersand replacement=\&, column sep=10pt, row sep=10pt, nodes={circle,inner sep=0pt,minimum size=4pt}},
  vSet/.style args = {(#1,#2)}{row #1 column #2/.style={nodes={draw,thin,fill=black!20!white}}},
  oldSet/.style args = {(#1,#2)}{row #1 column #2/.style={nodes={draw,thin,fill=black}}},
  oldSet/.list={(1,3),(2,1),(2,4),(1,2)},
  vSet/.list={(3,2),(3,3)},
]
\path[use as bounding box] (-1.70,-0.6) rectangle (0.7,0.6);
 \matrix [graph] (M)
 {
   {}\& {}\&[15pt] {}\& {}\& {}\&\\
   {}\& {}\& {}\& {}\& {}\&\\
   {}\& {}\& {}\& {}\& {}\&\\
 };
 \draw [thin]
   (M-1-2) -- (M-2-4)
   (M-1-2) -- (M-3-2)
   (M-1-2) -- (M-3-3)
   (M-1-3) -- (M-2-1)
   (M-1-3) -- (M-3-2)
   (M-1-3) -- (M-3-3)
   (M-2-1) -- (M-3-3)
   (M-2-4) -- (M-3-2);
\end{tikzpicture}
 \caption{$\overline{G}$.}
  \label{fig:ex1b}
 \end{subfigure}\qquad
 \begin{subfigure}[b]{.5\textwidth}
 \centering
\begin{tikzpicture}[
  graph/.style={matrix of math nodes, ampersand replacement=\&, column sep=30pt, row sep=14pt, nodes={circle,inner sep=0pt,minimum size=4pt}},
  vSet/.style args = {(#1,#2)}{row #1 column #2/.style={nodes={draw,thin,fill=black!20!white}}},
  vSet/.list={(1,1),(1,2),(1,3),(1,4),(1,5),(1,6),(2,1),(2,2),(2,3),(2,4),(2,5),(2,6)},
]
\path[use as bounding box] (-4.60,-1.0) rectangle (3.5, 1.3);
 \matrix [graph] (M)
 {
   {}\& {}\& {}\& {}\& {}\& {}\& \\
   {}\& {}\& {}\& {}\& {}\& {}\& \\
 };
 \foreach \i  in {1,2,3,4,5,6}{
   \draw [thin] (M-1-\i) -- (M-2-\i) {};
 }
%
 \foreach \i/\j in {1/2,2/3,3/4,4/5,5/6}{
   \draw [thin] (M-2-\i) -- (M-2-\j) {};
 }
 \draw [thin]
   (M-1-1) to[out=30,in=150]  (M-1-3)
   (M-1-1) to[out=35,in=135]  (M-1-4)
   (M-1-1) to[out=40,in=140]  (M-1-5)
   (M-1-2) to[out=30,in=150]  (M-1-4)
   (M-1-2) to[out=40,in=140]  (M-1-6)
   (M-1-3) to[out=30,in=150]  (M-1-5)
   (M-1-3) to[out=35,in=145]  (M-1-6)
   (M-1-4) to[out=30,in=150]  (M-1-6)
   (M-2-1) to[out=340,in=200] (M-2-6)
   (M-2-2) to[out=340,in=200] (M-2-5);
\node also [label=right:{\footnotesize $V(\bar{G})$}] (M-1-6) {};
\node also [label=right:{\footnotesize$V(G)$}] (M-2-6) {};
\end{tikzpicture}
  \caption{$G\overline{G}$.}
  \label{fig:ex1c}
 \end{subfigure}
\caption{Example of a graph, its complement and the resulting complementary prism.}
\label{fig:ex1}
\end{figure}

Haynes et al. \cite{Haynes-Henning-Slater-Merwe-2007} introduced the so-called \emph{complementary product} of two graphs as a generalization of the well-known \emph{Cartesian product} of the graphs. As a particular case of complementary products, the authors define the \emph{complementary prism} of a graph $G$, denoted by $G\overline{G}$, as the graph formed from the disjoint union of $G$ and its complement $\overline{G}$ by adding the edges of the perfect matching between the corresponding vertices  of $G$ and $\overline{G}$, where $V(G\overline{G})=V(G)\cup V(\overline{G})$. For the purposes of illustration, a graph $G$, its complement $\overline{G}$ and the complementary prism $G\overline{G}$ are depicted  respectively, in Figures \ref{fig:ex1a}, \ref{fig:ex1b} and \ref{fig:ex1c}. To simplify matters, $G$ and $\overline{G}$ are used to denote the subgraph copies of $G$ and $\overline{G}$, respectively, in $G\overline{G}$. 
For a set $X \subseteq V(G)$, let $\overline{X}$ denote the corresponding vertices of $X$ in $V(\overline{G})$.

\section{Complexity results}
\label{sec:pre}

Open-independent, open-locating-dominating sets were first studied by Seo and Slater \cite{Seo-Slater-2017} who stated necessary and sufficient conditions for the existence of an $OLD_{oind}$-set in a graph $G$ that has girth $g(G) \geq 5$. For general graphs (with arbitrary girth), the conditions stated in Theorem \ref{Seo-Sl-g(5)} are necessary but not sufficient, as is stated in Theorem \ref{OO-general}. We frequently use Theorem \ref{OO-general} in our proofs below, sometimes without mentioning the fact.

\begin{mytheorem}[\cite{Seo-Slater-2017}] \label{Seo-Sl-g(5)}
If a graph $G$ has girth satisfying $g(G) \geq 5$ and $S \subseteq V(G)$, then $S$ is an $OLD_{oind}$-set iff ($i$) each $v \in S$ is open-dominated by $S$ exactly once and ($ii$) each $v\notin S$ is open-dominated by $S$ at least twice.
\end{mytheorem}

\begin{mytheorem}[\cite{cappelle2019}] 
\label{OO-general}
If $S \subseteq V(G)$ is an $OLD_{oind}$-set of a graph $G$, then ($i$) each $v \in S$ is open-dominated exactly by $S$ once, and ($ii$) each $v\in V(G)\setminus S$ is open-dominated by $S$ at least twice.
\end{mytheorem}

Seo and Slater \cite{Seo-Slater-2017} presented some results about $OLD_{oind}$-sets in trees. The authors showed that every leaf and its neighbor are contained in any $OLD_{oind}$-set of any tree $T$, if $T$ has such a  set. Furthermore, they recursively defined the collection of trees that have unique $OLD_{oind}$-sets. On the other hand, 
they showed that \textsc{OLD-OIND} is an $\mathcal{NP}$-complete problem for general graphs. They present a reduction from the 3-SAT problem. Indeed, by reducing from the $\mathcal{NP}$-complete problem Planar 3-SAT, it is possible to prove $\mathcal{NP}$-completeness for planar graphs by using their construction. Recently, it was proved that \textsc{OLD-OIND} is $\mathcal{NP}$-complete for the complementary prisms of a given graph $G$ \cite{cappelle2020}. 
We use similar ideas to those in \cite{lu2002} for efficient edge domination, showing that \textsc{OLD-OIND} is $\mathcal{NP}$-complete even for the special cases when $G$ is either a planar bipartite graph of maximum degree five and girth six or a planar subcubic graph of girth nine, by polynomial reduction from the following decision problem.

\begin{flushleft}
    \fbox{
      \begin{minipage}{0.95\linewidth}
        \noindent {\textsc{Restricted X3C} (exact cover by 3-sets)}\\
        \textbf{Instance:} {A finite set $X$ with $|X| = 3n$ and a collection ${\cal S}$ of 3-element subsets of $X$ such that each element of $X$ is in at most 3 subsets, with $|{\cal S}| = m$, where $n$ and $m$ are positive integers.}\\
        \textbf{Question:} {Does ${\cal S}$ contain an exact cover of $X$, \textit{i.e.} a pairwise disjoint subcollection ${\cal S}'\subseteq{\cal S}$  such that every element of $X$ occurs in exactly one member of {${\cal S}'$?}}
      \end{minipage}}
    \medskip
  \end{flushleft}

It is well known that \textsc{Restricted X3C} is $\mathcal{NP}$-complete \cite{garey1982}. 
Each instance of \textsc{Restricted X3C}, say $X = \{x_1, x_2, \ldots,x_{3n}\}$ and ${\cal S} = \{S_1, S_2,\ldots,S_m\}$, can be associated with a bipartite graph $G_A = (V_A, E_A)$, where
$V_A =X \cup S$ and $E_A =\{(x_i, S_j) : 1 \leq i \leq 3n, 1 \leq j \leq m, \mbox{ and } x_i \in S_j\}$. If the associated bipartite graph $G_A$ is planar, the problem is termed the \emph{planar restricted exact cover by 3-sets problem} (\textsc{Planar Restricted X3C}), which is also $\mathcal{NP}$-complete \cite{dyer1986}.

\begin{mytheorem} \label{theo:npc1}
Deciding, for a given planar bipartite graph $G$ of maximum degree five and girth six, if $G$ has an $OLD_{oind}$-set is an $\mathcal{NP}$-complete problem.
\end{mytheorem}

\begin{proof}
It is possible to verify in polynomial time if a given set $D\subseteq V(G)$ is an $OLD_{oind}$-set of $G$. So, \textsc{OLD-OIND} is in $\mathcal{NP}$. We now show that \textsc{Planar Restricted X3C} is reducible in polynomial time to the problem \textsc{OLD-OIND} for planar bipartite graphs of maximum degree five and girth six. Let $X = \{x_1, x_2, \ldots,x_{3n}\}$ and ${\cal S} = \{S_1, S_2,\ldots,S_m\}$ be an instance of \textsc{Planar Restricted X3C}.
We now construct a graph $G = (V(G),E(G))$, as illustrated in Figure \ref{fig:OOPB}, where
$V(G) ={\{x_i : 1\leq i\leq3n\}}\cup{\{y_{ij} : 1 \leq i \leq 3n, 1 \leq j \leq m} \mbox{ and } x_i \in S_j\} \cup \{S_j,a_j,b_j,c_j,d_j,e_j : 1 \leq j \leq m\}$, and
$V(E) =\{x_iy_{ij},y_{ij}S_j : 1 \leq i \leq 3n, 1 \leq j \leq m, x_i \in S_j\}
 \cup\{S_ja_j,a_jb_j,b_jc_j,c_jd_j,d_je_j,e_jS_j : 1 \leq j \leq  m\}$. The resulting graph $G$ has $3n+9m$ vertices, and can be constructed in polynomial time. 

Clearly, $G$ is planar because $(X,S)$ is an instance of \textsc{Planar Restricted X3C}. Furthermore, $G$ is bipartite as $V(G)$ can be partitioned into the subsets: $\{x_i : 1\leq i\leq3n\}\cup \{S_j,b_j,d_j : 1 \leq j \leq m\}$ and $\{y_{ij} : 1 \leq i \leq 3n, 1 \leq j \leq m \mbox{ and } x_i \in S_j\}\cup \{a_j,c_j,e_j : 1 \leq j \leq m\}$.
It may also be noted that $G$ has maximum degree five, since each element of $X$ is in at most 3 subsets, and also it is easy to verify that $G$ has girth six. We are now going to prove that ($i$) $G$ has an $OLD_{oind}$-set if and only if ($ii$) there is a subset ${\cal S}' \subseteq {\cal S}$ that is an exact cover of $X$. 

To prove that ($i$) is necessary for ($ii$), assume that ${\cal S}$ has an exact cover ${\cal S}'$. Let $D=\{x_i,y_{ij},a_j,b_j,d_j,e_j\,:\,S_j \in {\cal S}' \mbox{ and } x_i \in S_j\}\cup\{S_j,a_j,c_j,d_j\,:\,S_j \in {\cal S}\setminus {\cal S}'\}$. Then $D$ contains $6n+4m$ vertices, which induce $3n+2m$ independent edges.  Note that every vertex $y_{ij}$ that is not in $D$ has its two neighbors in $D$, and the 6-cycle has four vertices in $D$. So it can be concluded that $D$ is an $OLD_{oind}$-set of $G$.

To prove that ($ii$) is necessary for ($i$), assume that $G$ has an $OLD_{oind}$-set $D$. Observe that $D$ does not contain both $S_j$ and $y_{ij}$, else $a_j,e_j \notin D$, which is impossible since by the open-independence of $D$, at most two of the vertices in the set $\{b_j,c_j,d_j\}$ are in $D$ and then there are two vertices of the 6-cycle $\langle S_j,a_j,b_j,c_j,d_j,e_j\rangle$ which are dominated exactly once by $D$. So $D$ is not an $OLD_{oind}$-set of $G$. The fact that $\{S_j,y_{ij}\}$ is not a subset of $D$ for all $j \in [m]$; implies that every $x_i\in D$ and it has a neighbor $y_{ij}$ for some $j\in [m]$ such that $y_{ij}\in D$. Note that altogether, these  results imply that for every $S_j$, the three vertices $y_{ij}$ are either (a) all in $D$, or (b) none are in $D$. If (a) then $S_j\notin D$ and $\{a_j,b_j,d_j,e_j\}\subseteq D$. If (b) then $S_j\in D$ and exactly one of the two symmetric subsets $\{a_j,c_j,d_j\}$ and $\{b_j,c_j,e_j\}$ are contained in $D$. Let ${\cal S}' =\{S_j  :  1 \leq j \leq m \mbox{ and } S_j \notin D\}$. Clearly, ${\cal S}'$ is an exact cover of ${\cal S}$.
\end{proof}

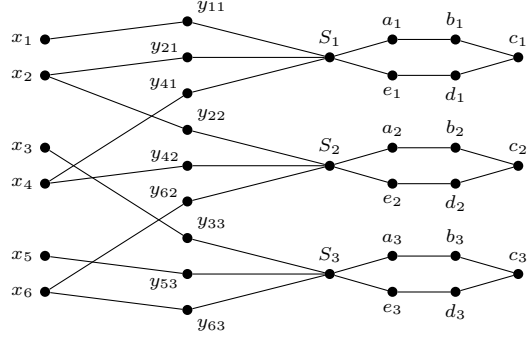
\begin{figure}
    \centering
    \begin{tikzpicture}[
  graph/.style={matrix of math nodes, ampersand replacement=\&, column sep=20pt, row sep=3pt, nodes={circle,inner sep=0pt,minimum size=3pt}},
  vSet/.style args = {(#1,#2)}{row #1 column #2/.style={nodes={draw,thick,fill=black!20!white}}},
  oldSet/.style args = {(#1,#2)}{row #1 column #2/.style={nodes={draw,thick,fill=black}}},
  every label/.append style={font=\scriptsize,label distance=-1pt},
  oldSet/.list={
   (1,2),
   (2,1),(2,4),(2,5),
   (3,2),(3,3),(3,6),
   (4,1),(4,4),(4,5),
   (5,2),
   (6,2),
   (7,1),(7,4),(7,5),
   (8,2),(8,3),(8,6),
   (9,1),(9,4),(9,5),
   (10,2),
   (11,2),
   (12,1),(12,4),(12,5),
   (13,2),(13,3),(13,6),
   (14,1),(14,4),(14,5),
   (15,2)
  },
]
 \matrix [graph] (M) 
 {
   {}\& [30pt] {}\& [30pt]{}\& {}\& {}\& {}\\
   {}\& {}\& {}\& {}\& {}\& {}\\
   {}\& {}\& {}\& {}\& {}\& {}\\
   {}\& {}\& {}\& {}\& {}\& {}\\
   {}\& {}\& {}\& {}\& {}\& {}\\[7pt]
   {}\& {}\& {}\& {}\& {}\& {}\\
   {}\& {}\& {}\& {}\& {}\& {}\\
   {}\& {}\& {}\& {}\& {}\& {}\\
   {}\& {}\& {}\& {}\& {}\& {}\\
   {}\& {}\& {}\& {}\& {}\& {}\\[7pt]
   {}\& {}\& {}\& {}\& {}\& {}\\
   {}\& {}\& {}\& {}\& {}\& {}\\
   {}\& {}\& {}\& {}\& {}\& {}\\
   {}\& {}\& {}\& {}\& {}\& {}\\
   {}\& {}\& {}\& {}\& {}\& {}\\[7pt]
 };
 \foreach \i/\j in {2/1,4/3,4/6,7/11,9/5,9/8,12/13,14/10,14/15}
  \draw [thin] (M-\i-1) -- (M-\j-2);
 \foreach \i/\j in {3/1,3/3,3/5,8/6,8/8,8/10,13/11,13/13,13/15}
  \draw [thin] (M-\i-3) -- (M-\j-2);
 \foreach \i/\j in {3/2,3/4,8/7,8/9,13/12,13/14}
  \draw [thin] (M-\i-3) -- (M-\j-4)
               (M-\i-6) -- (M-\j-5) -- (M-\j-4);
 \foreach[count=\c] \i in {2,4,7,9,12,14}
  \node also [label={180:$x_\c$}] (M-\i-1) {};
 \foreach \i/\l/\a/\d in 
 {
 1/11/70/-2,
 3/21/110/-4,
 5/41/100/-4,
 6/22/70/-2,
 8/42/110/-4,
 10/62/100/-4,
 11/33/70/-2,
 13/53/260/-4,
 15/63/280/-2}
 \node also [label={[label distance =\d pt]\a:$y_{\l}$}] (M-\i-2) {};
 \foreach[count=\c,evaluate=\i as \j using int(\i-1),evaluate=\i as \k using int(\i+1)] \i in {3,8,13}{
  \node also [label={90:$S_\c$}] (M-\i-3) {};
  \node also [label={90:$a_\c$}] (M-\j-4) {};
  \node also [label={90:$b_\c$}] (M-\j-5) {};
  \node also [label={90:$c_\c$}] (M-\i-6) {};
  \node also [label={270:$d_\c$}] (M-\k-5) {};
  \node also [label={270:$e_\c$}] (M-\k-4) {};
 }
\end{tikzpicture}
    \caption{Graph $G$ that is an instance of \textsc{OLD-OIND} for planar bipartite graphs of maximum degree five and girth six as in the proof of Theorem \ref{theo:npc1}). ${\cal S} = \{S_1, S_2, S_3\} = \{\{x_1, x_2, x_4\}, \{x_2, x_4, x_6\}, \{x_3, x_5, x_6\}\}$.}
    \label{fig:OOPB}
\end{figure}

On the other hand, if odd cycles are allowed, \textit{i.e.}, the graph can be non-bipartite, we can modify the construction in the proof of Theorem \ref{theo:npc1} to add the restriction that the planar graph is subcubic. In this case, the instances have girth nine.

\begin{mytheorem}  \label{theo:npc2}
Deciding, for a given planar subcubic graph $G$ of girth nine, if $G$ has an $OLD_{oind}$-set is an $\mathcal{NP}$-complete problem.
\end{mytheorem}
\begin{proof}
It is possible to verify in polynomial time if a given set $D\subseteq V(G)$ is an $OLD_{oind}$-set of $G$. So, \textsc{OLD-OIND} is in $\mathcal{NP}$. We now show that, for planar subcubic graphs of girth nine, \textsc{Planar Restricted X3C} is reducible in polynomial time to the problem \textsc{OLD-OIND}. Let $X = \{x_1, x_2, \ldots,x_{3n}\}$ and ${\cal S} = \{S_1, S_2,\ldots,S_m\}$ be an instance of \textsc{Planar Restricted X3C}.
We now construct a graph $G = (V(G),E(G))$, as illustrated in Figure \ref{fig:OOX3C}. Let $X=\{x_1,\ldots,x_{3n}\}$, $Y=\{y_{ij} : 1 \leq i \leq 3n, 1 \leq j \leq m \mbox{ and } x_i \in S_j\}$, and let $C_j=\{a_j,b_j,c_j,d_j,e_j,f_j,g_j,h_j,k_j\}$, for every $j \in [m]$.
Let $V(G)$ be the set containing  $X\cup Y$ and the set $C_j$, for every $j \in [m]$. Add the edges that form a 9-cycle among the vertices in $C_j$, for every $j \in [m]$; add the edges of a perfect matching between the three vertices $y_{ij}$ and the vertices in the set $D^*_j=\{c_j,k_j,f_j\}$, for every $j \in [m]$; and finally, add the edges in $\{x_iy_{ij} : 1 \leq i \leq 3n, 1 \leq j \leq m \mbox{ and } x_i \in S_j\}$. The resulting graph $G$ has $3n+12m$ vertices.

Clearly, $G$ is planar because $(X,S)$ is an instance of \textsc{Planar Restricted X3C}. Furthermore, $G$ has maximum vertex degree at most 3, since each vertex $x_i$ has at most three neighbors in $Y$, and it is easy to verify that $G$ has girth nine. We are now going to prove that ($i$) $G$ has an $OLD_{oind}$-set if and only if ($ii$) there is a subset ${\cal S}' \subseteq {\cal S}$ that is an exact cover of $X$. 

For every $j \in [m]$, consider the sets $D^1_j=\{a_j,b_j,d_j,e_j,g_j,h_j\}$, $D^2_j=\{b_j,e_j,h_j\}$ and $D^3_j=\{a_j,d_j,g_j\}$.
To prove that ($i$) is necessary for ($ii$), assume that ${\cal S}$ has an exact cover ${\cal S}'$. 
Let $D$ be the set containing, the vertices in $D^1_j \cup\{x_i,y_{ij}\,:\, x_i \in S_j\}$, for every $j$ such that $S_j \in {\cal S}'$; and, additionally, the vertices in $D^*_j \cup D^2_j$, for every $j$ such that $S_j \notin {\cal S}'$. Then $D$ contains $6n+6m$ vertices, which induce $3n+3m$ independent edges.  Note that every vertex in every cycle induced by $C_j$ has two neighbors in $C_j \cap D$ and  every vertex $y_{ij}$ such that $S_j \notin {\cal S}'$ is not in $D$ has its two neighbors in $D$. So, it can be concluded that $D$ is an $OLD_{oind}$-set of $G$.

To prove that ($ii$) is necessary for ($i$), assume that $G$ has an $OLD_{oind}$-set $D$. First, we prove some claims.

\begin{claim} \label{claim1}
In every cycle induced by $C_j$, either a) vertices $k_j$, $c_j$ and $f_j$ are all in $D$, or b) none are in $D$.
\end{claim}
\noindent\textbf{Proof of Claim \ref{claim1}:} By inspection, we can verify that exactly six vertices of $C_j$ are in any $OLD_{oind}$-set of $G$. So the result follows by symmetry. $\Box$

\begin{claim}\label{claim2}
The set $D$ does not contain both $y_{ij}$ and a vertex in $D^*_j$, for all $j \in [m]$. 
\end{claim}
\noindent\textbf{Proof of Claim \ref{claim2}:} By contradiction, suppose that $D$ contains, for some $j \in [m]$, a vertex $y_{ij}$ and its neighbor in $D^*_j$. This implies, by Claim \ref{claim1}, that $D^*_j$ is a subset of $D$. Since $D$ is open-independent, $D \cap C_j$ has exactly three vertices, and there are six vertices of $C_j$ which are dominated by $D$ exactly once. So $D$ is not an $OLD_{oind}$-set of $G$.
$\Box$

By Claim \ref{claim2}, $D$ does not contain both $y_{ij}$ and a vertex in $D^*_j$. This implies that $X\subseteq D$ and every $x_i$ has a unique neighbor $y_{ij}$ for some $j\in [m]$ such that $y_{ij}\in D$. By Claim \ref{claim1}, for every $j \in [m]$, the three vertices $k_j$, $c_j$ and $f_j$ are either (a) all in $D$, or (b) none are in $D$. If (a) then   exactly one of the two symmetric subsets $D^2_j$ and $D^3_j$ are contained in $D$. If (b) then $D^1_j\subseteq D$. Let ${\cal S}' =\{S_j  :  1 \leq j \leq m \mbox{ and } D^*_j \cap D=\emptyset\}$. Clearly, ${\cal S}'$ is an exact cover of ${\cal S}$.
\end{proof}

\begin{figure}
    \centering
    \begin{tikzpicture}[
  graph/.style={matrix of math nodes, ampersand replacement=\&, column sep=20pt, row sep=4pt, nodes={circle,inner sep=0pt,minimum size=3pt}},
  vSet/.style args = {(#1,#2)}{row #1 column #2/.style={nodes={draw,thick,fill=black!20!white}}},
  oldSet/.style args = {(#1,#2)}{row #1 column #2/.style={nodes={draw,thick,fill=black}}},
  every label/.append style={font=\scriptsize,label distance=-1pt},
  oldSet/.list={
   (1,2),
   (2,1),(2,4),(2,5),(2,6),(2,7),
   (3,2),(3,3),
   (4,4),(4,5),(4,6),(4,7),
   (5,1),(5,2),
   (6,2),
   (7,1),(7,4),(7,5),(7,6),(7,7),
   (8,2),(8,3),
   (9,4),(9,5),(9,6),(9,7),
   (10,1),(10,2),
   (11,2),
   (12,1),(12,4),(12,5),(12,6),(12,7),
   (13,2),(13,3),
   (15,1),(14,4),(14,5),(14,6),(14,7),
   (15,2)
  },
]
 \matrix [graph] (M) 
 {
   {}\& [30pt] {}\& [30pt]{}\& {}\& {}\& {}\& {}\\[7pt]
   {}\& {}\& {}\& {}\& {}\& {}\& {}\\
   {}\& {}\& {}\& {}\& {}\& {}\& {}\\
   {}\& {}\& {}\& {}\& {}\& {}\& {}\\[7pt]
   {}\& {}\& {}\& {}\& {}\& {}\& {}\\[7pt]
   {}\& {}\& {}\& {}\& {}\& {}\& {}\\[7pt]
   {}\& {}\& {}\& {}\& {}\& {}\& {}\\
   {}\& {}\& {}\& {}\& {}\& {}\& {}\\
   {}\& {}\& {}\& {}\& {}\& {}\& {}\\[7pt]
   {}\& {}\& {}\& {}\& {}\& {}\& {}\\[7pt]
   {}\& {}\& {}\& {}\& {}\& {}\& {}\\[7pt]
   {}\& {}\& {}\& {}\& {}\& {}\& {}\\
   {}\& {}\& {}\& {}\& {}\& {}\& {}\\
   {}\& {}\& {}\& {}\& {}\& {}\& {}\\[7pt]
   {}\& {}\& {}\& {}\& {}\& {}\& {}\\[7pt]
 };
 \foreach \i/\j in {2/1,5/3,5/6,7/11,10/5,10/8,12/13,15/10,15/15}
  \draw [thin] (M-\i-1) -- (M-\j-2);
 \foreach \i/\j in {3/3,8/8,13/13}
  \draw [thin] (M-\i-3) -- (M-\j-2);
%

 \foreach \i/\j in {2/1,7/6,12/11}
  \draw [thin] (M-\i-6) to[out=140,in=355] (M-\j-2);
 \foreach \i/\j in {4/5,9/10,14/15}
  \draw [thin] (M-\i-6) to[out=220,in=5] (M-\j-2);
 \foreach \i/\j in {3/2,3/4,8/7,8/9,13/12,13/14}
  \draw [thin] (M-\i-3) -- (M-\j-4)
               (M-\j-7) -- (M-\j-6) -- (M-\j-5) -- (M-\j-4);
  \draw [thin] (M-2-7) -- (M-4-7)
               (M-7-7) -- (M-9-7)
               (M-12-7) -- (M-14-7);
 \foreach[count=\c] \i in {2,5,7,10,12,15}
  \node also [label={180:$x_\c$}] (M-\i-1) {};
 \foreach \i/\l/\a/\d in 
 {
 1/11/70/-2,
 3/21/110/-4,
 5/41/100/-4,
 6/22/70/-2,
 8/42/110/-4,
 10/62/100/-4,
 11/33/70/-2,
 13/53/260/-4,
 15/63/280/-2}
 \node also [label={[label distance =\d pt]\a:$y_{\l}$}] (M-\i-2) {};
 \foreach[count=\c,evaluate=\i as \j using int(\i-1),evaluate=\i as \k using int(\i+1)] \i in {3,8,13}{
  \node also [label={90:$k_\c$}] (M-\i-3) {};
  \node also [label={90:$a_\c$}] (M-\j-4) {};
  \node also [label={90:$b_\c$}] (M-\j-5) {};
  \node also [label={90:$c_\c$}] (M-\j-6) {};
  \node also [label={90:$d_\c$}] (M-\j-7) {};
  \node also [label={270:$e_\c$}] (M-\k-7) {};
  \node also [label={270:$f_\c$}] (M-\k-6) {};
  \node also [label={270:$g_\c$}] (M-\k-5) {};
  \node also [label={270:$h_\c$}] (M-\k-4) {};
 }
\end{tikzpicture}
    \caption{Graph $G$ that is an instance of \textsc{OLD-OIND} for planar subcubic graphs of girth nine as in the proof of Theorem \ref{theo:npc2}. ${\cal S} = \{S_1, S_2, S_3\} = \{\{x_1, x_2, x_4\}, \{x_2, x_4, x_6\}, \{x_3, x_5, x_6\}\}$.}
    \label{fig:OOX3C}
\end{figure}
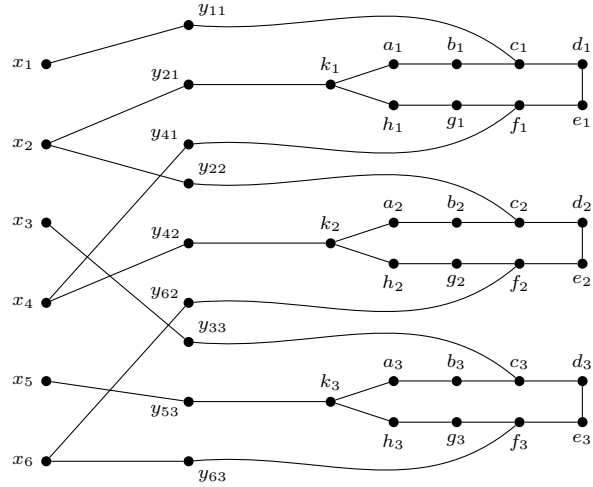


\section{{$P_4$}-tidy graphs}
\label{sec:p4-tidy}


In this section, we study $OLD_{oind}$-sets in $P_4$-tidy graphs. A graph $G = (V,E)$ is termed \emph{$P_4$-tidy} if, for every vertex set $A$ inducing a $P_4$ in $G$, there is at most one vertex $v \in V \setminus A$ such that $G[A \cup \{v\}]$ contains at least two induced $P_4$'s. This class includes spider and quasi-spider graphs, which we now define.

A \emph{spider} is a graph whose vertex set has a partition $(C, X, H)$, where $C = \{c_1,\ldots, c_k\}$ and $X = \{x_1,\ldots, x_k\}$ for a given integer $k \geq 2$ are respectively, a clique and an independent set; $x_i$ is adjacent to $c_j$ iff $i = j$ (a thin spider), or $x_i$ is adjacent to $c_j$ if and only if $i \neq j$ (a thick spider); and every vertex of $H$ is adjacent to each vertex of $C$ and is not adjacent to any vertex of $X$. The size $k$ of both $C$ and $X$ is called the \emph{weight} of the spider and the set $H$ in the partition is called its \emph{head}. Notice that if $k=2$, then the thick and thin spider graphs are isomorphic. A \textit{quasi-spider} is a graph obtained from a spider that has vertex partition $(C, X, H)$ by replacing at most one vertex of $C \cup X$ by a $K_2$ or a $\overline{K_2}$ (where each vertex of the $K_2$ or the $\overline{K_2}$ has the same adjacency structure as the vertex it replaced). 

\begin{figure}[!ht]
\centering
 \begin{subfigure}[t]{.20\textwidth}
 \centering
  \begin{tikzpicture}[
graph/.style={matrix of math nodes, ampersand replacement=\&, column sep=10pt, row sep=10pt, nodes={circle,inner sep=0pt,minimum size=3pt}},
vSet/.style args = {(#1,#2)}{row #1 column #2/.style={nodes={draw,thick,fill=black!20!white}}},
oldSet/.style args = {(#1,#2)}{row #1 column #2/.style={nodes={draw,thick,fill=black}}},
oldSet/.list={(1,2),(1,4),(2,1),(2,3),(2,5),(3,1),(3,3),(3,5)},
]
\path[use as bounding box] (-1.0,-0.5) rectangle (1.0, 0.95);
\matrix [graph] (M)
{
 {}\& {}\& {}\& {}\& {}\\
 {}\& {}\& {}\& {}\& {}\\
 {}\& {}\& {}\& {}\& {}\\
};
\draw [thin]
(M-2-1) to [out=90,in=90,looseness=1.6] (M-2-5)
(M-1-2) -- (M-1-4)
(M-1-2) -- (M-2-1)
(M-1-2) -- (M-2-3)
(M-1-2) to [out=40,in=110,looseness=1.3] (M-2-5)
(M-1-4) to [out=140,in=70,looseness=1.3] (M-2-1)
(M-1-4) -- (M-2-3)
(M-1-4) -- (M-2-5)
(M-2-1) -- (M-2-3)
(M-2-1) -- (M-3-1)
(M-2-3) -- (M-3-3)
(M-2-3) -- (M-2-5)
(M-2-5) -- (M-3-5);
\end{tikzpicture}
  \caption{{A thin spider of weight three with a nonempty head.}}
  \label{fig:spider1}
 \end{subfigure}\qquad
 \begin{subfigure}[t]{.20\textwidth}
 \centering
 \begin{tikzpicture}[
graph/.style={matrix of math nodes, ampersand replacement=\&, column sep=5.4pt, row sep=25pt, nodes={circle,inner sep=0pt,minimum size=3pt}},
vSet/.style args = {(#1,#2)}{row #1 column #2/.style={nodes={draw,thick,fill=black!20!white}}},
oldSet/.style args = {(#1,#2)}{row #1 column #2/.style={nodes={draw,thick,fill=black}}},
oldSet/.list={(1,1),(1,3),(1,5),(1,7),(2,2),(2,4),(2,6)},
]
\matrix [graph] (M)
{
 {}\& {}\& {}\& {}\& {}\& {}\& {}\\ 
 {}\& {}\& {}\& {}\& {}\& {}\& {}\\
};
 \foreach \i/\j in {1/4,1/6,3/4,3/6,5/2,5/6,7/2,7/4}
  \draw [thin] (M-1-\i) -- (M-2-\j);
 \foreach \i [evaluate=\i as \j using int(\i+2)] in {1,3,5} 
  \draw [thin] (M-1-\i) -- (M-1-\j);
 \draw [thin]
  (M-1-1) to[out=20,in=140]  (M-1-5)
  (M-1-1) to[out=40,in=140]  (M-1-7)
  (M-1-3) to[out=40,in=160]  (M-1-7);
\end{tikzpicture}
 \caption{{A thick quasi-spider of weight three.} }
  \label{fig:spider2}
 \end{subfigure}\qquad
  \begin{subfigure}[t]{.20\textwidth}
 \centering
 \begin{tikzpicture}[
  graph/.style={matrix of math nodes, ampersand replacement=\&, column sep=5pt, row sep=25pt, nodes={circle,inner sep=0pt,minimum size=3pt}},
  vSet/.style args = {(#1,#2)}{row #1 column #2/.style={nodes={draw,thick,fill=black!20!white}}},
  oldSet/.style args = {(#1,#2)}{row #1 column #2/.style={nodes={draw,thick,fill=black}}},
  oldSet/.list={(1,2),(1,4),(1,5),(2,1),(2,3),(2,4),(2,5)},
]
 \matrix [graph] (M)
 {
   {}\& {}\& {}\& [10pt]{}\& [10pt]{}\\ 
   {}\& {}\& {}\& {}\& {}\\
 };
 \draw [thin]
   (M-1-2) -- (M-2-1)
   (M-1-2) -- (M-2-3)
   (M-1-2) -- (M-1-4)
   (M-1-4) -- (M-2-4)
   (M-1-4) -- (M-1-5)
   (M-1-5) -- (M-2-5);
 \draw [thin]
   (M-1-2) to[out=30,in=150]  (M-1-5);
\end{tikzpicture}
 \caption{{A thin quasi-spider of weight three.}}
  \label{fig:spider3}
 \end{subfigure}\qquad
   \begin{subfigure}[t]{.25\textwidth}
 \centering
 \begin{tikzpicture}[
  graph/.style={matrix of math nodes, ampersand replacement=\&, column sep=5pt, row sep=25pt, nodes={circle,inner sep=0pt,minimum size=3pt}},
  vSet/.style args = {(#1,#2)}{row #1 column #2/.style={nodes={draw,thick,fill=black!20!white}}},
  oldSet/.style args = {(#1,#2)}{row #1 column #2/.style={nodes={draw,thick,fill=black}}},
  oldSet/.list={(1,2),(1,4),(2,1),(2,3),(2,4)},
]
 \matrix [graph] (M)
 {
   {}\& {}\& {}\& [10pt]{}\\ 
   {}\& {}\& {}\& {}\\
 };
 \draw [thin]
   (M-1-2) -- (M-2-1)
   (M-1-2) -- (M-2-3)
   (M-1-2) -- (M-1-4)
   (M-1-4) -- (M-2-4)
   (M-2-1) -- (M-2-3);
\node also [label=right:{\scriptsize $Z$}] (M-1-4) {};
\end{tikzpicture}
 \caption{{The $Z$ graph: a thick and thin quasi-spider of weight two.}}
  \label{fig:spider4}
 \end{subfigure}\qquad
\caption{Examples of spider and quasi-spider graphs.}
\label{fig:ex-spiders}
\end{figure}
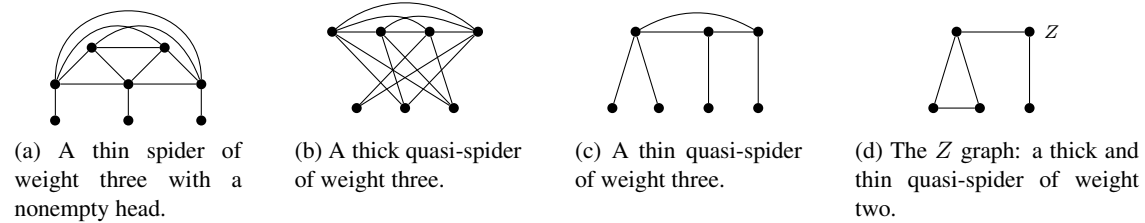

The following is a structural theorem for $P_4$-tidy graphs in terms of spider and quasi-spider graphs. Spiders and quasi-spiders are co-connected graphs.

\begin{mytheorem} [\cite{giakoumakis1997}] \label{p4-tidy}
$G$ is a $P_4$-tidy graph iff if exactly one of the following statements holds:
\begin{enumerate} [topsep=0pt,itemsep=-1ex,partopsep=1ex,parsep=1ex]
    \item $G$ is the union or the join of two $P_4$-tidy graphs;
    \item $G$ is a spider or a quasi-spider graph with partition $(C, X, H)$ such that either $H$ induces a $P_4$-tidy graph or is empty;
    \item $G$ is isomorphic to $C_5$, $P_5$, $\overline{P_5}$, or $K_1$.
\end{enumerate}
\end{mytheorem}

The following two theorems settle which spiders and quasi-spiders have an $OLD_{oind}$-set. Let the graph in Figure \ref{fig:spider4} be denoted by $Z$. 

\begin{restatable}{mytheorem}{spidernotOO}
\label{spider-not-OO}
No spider has an $OLD_{oind}$-set.
\end{restatable}
\begin{proof}
Let $G = (V(G),E(G))$ be a spider with vertex set $V(G)$ partitioned into $(C,X,H)$, where $|C|=|X| =k \geq2$. For a proof by contradiction, suppose that $G$ has an $OLD_{oind}$-set $S \subseteq V(G)$. We consider the subgraph $G[C \cup X]$ and show that either $S$ is not an open-independent set or $S$ is not an open-dominating set. As the cases $H\neq \emptyset$ and $H=\emptyset$ are analogous, they are not considered further in the following proof. If $G$ is thin it contains at least two pendant vertices in $X$ say, $x_i$ and $x_j$. If both $x_i,x_j \in S$, then both $c_i,c_j \in S$. As $C$ is a clique, $c_i$ and $c_j$ are adjacent in $C$ and thus, $c_i$($c_j$) is open-dominated by both $x_i$ and $c_j$($x_j$ and $c_i$). By Theorem \ref{OO-general}(i), $S$ cannot be an $OLD_{oind}$-set. If at least one of $x_i,x_j \notin S$, then that vertex cannot be open-dominated by $S$ at least twice. By Theorem \ref{OO-general}(ii), $S$ cannot be an $OLD_{oind}$-set. Suppose instead that $G$ is thick. If $k=2$, then $G$ is isomorphic to a thin spider and the result follows. So, assume that $k\geq3$. As $C$ is a clique, $|C \cap S| \leq 2$, and the following three subcases demonstrate that $S$ is not open-dominating. 
\begin{enumerate} [topsep=0pt,itemsep=-1ex,partopsep=1ex,parsep=1ex]
  \item If $C \cap S = \emptyset$, each $x_i \in X$ must belong to $S$, but no $x_i$ is open-dominated by any vertex in $S$. 

 \item If $|C \cap S| = 1$, say $c_i \in S$ for a unique $i \in [k]$ then, as $G$ is thick, $x_i$ is not open-dominated by $c_i$. Whether or not $x_i \in S$, $x_i$ is not open-dominated by any member of $S$.

 \item If $|C \cap S| = 2$, say $c_i, c_j \in S$, where $i, j \in [k]$,  then $x_i$ (resp. $x_j$) is not a member of $S$ and is open-dominated by only $c_j$ (resp. $c_j$). Hence, both $x_i$ and $x_j$ are open-dominated exactly once each by $S$, and thus $S$ is not an $OLD_{oind}$-set. 
  \end{enumerate}
\end{proof}

\begin{restatable}{mytheorem}{quasispidernotOO}
\label{quasi-spider-not-OO}
Any quasi-spider that has an $OLD_{oind}$-set is isomorphic to the thick and thin quasi-spider $Z$ depicted in Fig. \ref{fig:spider4}.  
\end{restatable}
\begin{proof}
Let $G = (V(G),E(G))$ be a quasi-spider obtained from a given spider $G_s$ that has a vertex partition $(C,X,H)$, where once again, $|C|=|X| =k \geq2$. Suppose $G$ is produced by applying the vertex replacement operation to a vertex $v \in C\cup X$ in $G_s$ and that $G$ has an $OLD_{oind}$-set $S \subseteq V(G)$. Again, we consider the subgraph $G[C \cup X]$ and show that unless $k=2$, either $S$ is not an open-locating-dominating set or is not an open-independent set of $G$ and thus cannot be an $OLD_{oind}$-set. 

First we consider $k=2$. If $H=\emptyset$, since headless thick and thin spiders of size two are isomorphic, then $G_s$ is a 4-path, and there are four possibilities for $G$. In the first possibility, where either vertex in $X$ is replaced by a $K_2$, $G$ is isomorphic to $Z$ and thus $G$ has an $OLD_{oind}$-set. In the other three possibilities, where a vertex in $X$ is replaced by a $\overline{K_2}$, or a vertex in $C$ is replaced by a $K_2$ or a $\overline{K_2}$; it is easy to establish by inspection that $G$ does not have an $OLD_{oind}$-set. 
If $H\neq\emptyset$, then $G$ can possibly have an $OLD_{oind}$-set only if $H$ has a dominating subset. If this is so there are eight cases - where $G_s$ is either thick or thin and a vertex in either $C$ or $X$ is replaced by either a $K_2$ or a $\overline{K_2}$. It is straightforward to establish by inspection that none of the cases has an an $OLD_{oind}$-set.
 


Now, we consider $k\geq 3$. Whether $G_s$ is thick or thin, there are four cases, depending on which of the following vertex replacement operations is applied: either $v \in C$ or $v \in X$ is replaced by either a $K_2$ or a $\overline{K_2}$. Let $c^{\ell}_i$ (resp. $x^{\ell}_i$), $\ell=1,2$, be the two vertices that replace vertex $v$ in $C$ (resp. $X$). When no index $\ell$ is indicated, $c_i$ (resp. $x_i$) can be either of the two vertices $c^1_i$ and $c^2_i$ (resp. $x^1_i$ and $x^2_i$). Replacing a vertex in $C$ induces in $G$ either a clique or a clique with one edge missing, denoted in either case by $\overline{C}$. As $S$ must be open-independent to be an $OLD_{oind}$-set, it follows that $|\overline{C}\cap S|\leq 2$. Instead, replacing a vertex in $X$ induces in $G$ either an independent set or an independent set with one additional edge, denoted in either case by $\overline{X}$.

If $G_s$ is thin, after any of the four types of replacement, the pair of vertices in $G$ that replace $v$ cannot be distinguished and thus, in this case, $G$ does not have an an $OLD_{oind}$-set. 
If $G_s$ is thick, similar arguments to those used in the proof of Theorem \ref{spider-not-OO} for thick spiders can be applied to establish that, once again, $G$ does not have an $OLD_{oind}$-set, as follows:

\begin{enumerate} [topsep=0pt,itemsep=-1ex,partopsep=1ex,parsep=1ex]
  \item \label{teo4-i} If $\overline{C} \cap S = \emptyset$ then, unless a vertex $x_i \in X$ is replaced by a $K_2$, no vertex in $X$ is open-dominated by any vertex in $S$. If indeed, $x_i$ is replaced by a $K_2$, there are still at least two vertices in $X$ that are not open-dominated by any vertex in $S$. 

 \item \label{teo4-ii} If $|\overline{C} \cap S| = 1$, say $c_i\in S$ for a unique $i \in [k]$ then, as $G$ is thick, $x_i$ is not open-dominated by $c_i$. Whether or not $x_i \in S$, $x_i$ is not open-dominated by any member of $S$.

 \item \label{teo4-iii} If $|\overline{C} \cap S| = 2$, say $c_i, c_j \in S$, where $i, j \in [k]$ (with $i=j$ when $c^1_i,c^2_i \in S$), then $x_i$ (resp. $x_j$) is not a member of $S$ and is open-dominated by only $c_j$ (resp. $c_j$). Hence, both $x_i$ and $x_j$ are open-dominated exactly once each by $S$, and thus $S$ is not an $OLD_{oind}$-set.
\end{enumerate}
\end{proof}

We now consider the existence of an $OLD_{oind}$-set in a $P_4$-tidy graph.
If $G$ is a disconnected $P_4$-tidy graph having an $OLD_{oind}$-set, then each component of $G$ has an $OLD_{oind}$-set. Thus we may consider only connected $P_4$-tidy graphs.

\begin{restatable}{mytheorem}{ptidyfirst}
\label{p4-tidy-1}
Let $G$ be a co-connected $P_4$-tidy graph. $G$ has an $OLD_{oind}$-set iff $G$ is isomorphic to either $P_5$ or to the graph $Z$ shown in Figure \ref{fig:spider4}.  
\end{restatable}
\begin{proof}
It can be verified by inspection that $P_5$ and $Z$ are co-connected $P_4$-tidy graphs having $OLD_{oind}$-sets.
Let $G$ be a co-connected $P_4$-tidy graph. By Theorem \ref{p4-tidy}, $G$ is either a spider, a quasi-spider or one of the graphs $C_5$, $P_5$ and $\overline{P_5}$. 
Suppose that $G$ has an $OLD_{oind}$-set.
Clearly, $G$ is not isomorphic to either $C_5$ or $\overline{P_5}$. By Theorems \ref{spider-not-OO} and \ref{quasi-spider-not-OO} $G$ is isomorphic to the graph $Z$.
\end{proof}

\begin{restatable}{mytheorem}{ptidysecond}
\label{p4-tidy-2}
Let $G$ be a connected $P_4$-tidy graph of order $n\geq 2$. $G$ has an $OLD_{oind}$-set iff either $G$ is isomorphic to one of the graphs $K_2$, $K_3$, $P_5$, $P_5 \bowtie K_1$, $Z$, and $Z \bowtie K_1$ or it can be obtained from them recursively by applying the following operation. Let $t\geq 2$ and $G_1, \ldots, G_t$ be connected $P_4$-tidy graphs each having an $OLD_{oind}$-set. Set $G:=(G_1 \oplus \cdots \oplus G_t)\bowtie K_1$. 
\end{restatable}
\begin{proof}
Clearly, $K_2$, $K_3$, $P_5$, $P_5 \bowtie K_1$, $Z$, and $Z \bowtie K_1$ are connected $P_4$-tidy graphs, each having an $OLD_{oind}$-set. Suppose instead that $G \notin \{K_2, K_3, P_5, P_5 \bowtie K_1,Z, Z \bowtie K_1\}$. Let $G_1,\ldots,G_t$ be graphs such that $G_i$ has an $OLD_{oind}$-set $S_i$, for each $i\in [t]$ and let $G$ be a graph obtained by the described operation. Let $S=\bigcup_{i\in[t]}S_i$ and $V(K_1)=\{v\}$. It is easy to see that every vertex in $V(G)\setminus\{v\}$ is both distinguished and open-dominated by $S$, and that $v$ is the only vertex of $G$ that is adjacent to all of the vertices of $S$. Hence $S$ is also an $OLD_{oind}$-set of $G$.  

Conversely, suppose that $G$ is a connected $P_4$-tidy graph distinct from $K_2$ and $K_3$ such that $G$ has an $OLD_{oind}$-set $S$. Hence $n\geq4$. If $\overline{G}$ is connected, $G$ is co-connected and by Theorem \ref{p4-tidy-1}, $G$ is isomorphic to one of $P_5$ and $Z$. If $\overline{G}$ is disconnected, $G$ can be obtained by the join of two graphs, say $H_1$ and $H_2$. Firstly, we claim that one of $H_1$ and $H_2$ has exactly one vertex $v$ say, and $v\notin S$. To the contrary, suppose $S_1=S\cap V(H_1)\neq \emptyset$ and $S_2=S\cap V(H_2)\neq \emptyset$. In order to be open-independent $|S_1|=|S_2|=1$, and $|S|=2$, which implies that $G$ has at most three vertices.  Thus, one of $H_1$ and $H_2$ has an empty intersection with $S$. Without loss of generality, we may assume that $V(H_2)\cap S=\emptyset$. If $H_2$ has at least two vertices $u$ and $v$ say, then $N_G(u)\cap S=N_G(v)\cap S$ which contradicts the fact that $S$ is an $OLD$-set of $G$. So $H_2$ is isomorphic to $K_1$ and, $S=S_1$ is an $OLD_{oind}$-set of $H_1$. 

We now analyze the possibilities for $H_1$. First suppose that $H_1$ is connected. If $|S|=2$, $G$ has at most three vertices. Since $n\geq 4$, $|S|\geq 4$. Observe that, if $H_1$ is the join of two graphs, the vertices of $S$ must all belong to the same subgraph and any vertex of any other subgraph is adjacent to every vertex of $S$. Altogether, these facts imply that there exists a vertex in $H_1$, $y$ say, such that $S\subseteq N_{H_1}(y)$. In this case, $v$ and $y$ are not distinguished by $S$ in $G$. Thus we may assume that $H_1$ is co-connected. By Theorem \ref{p4-tidy-1}, $H_1$ is isomorphic either to $P_5$ ($G=P_5\bowtie K_1$) or $Z$ ($G=Z\bowtie K_1$). 
Now instead, assume that $H_1$ is disconnected. In this case $H_1=G_1,\ldots,G_t$ is a collection of $t\geq 2$ connected $P_4$-tidy graphs. As every vertex in $V(G_i)$ has exactly one additional neighbor in $G$ (\textit{i.e.} vertex $v\notin S$), if there exists $i\in [t]$ such that $G_i$ does not have an $OLD_{oind}$-set, it is easy to see that $S$ is not an $OLD_{oind}$-set of $G$. Thus, we can conclude that every $G_i$ with $i\in [t]$ has an $OLD_{oind}$-set and that $G$ is the graph $H_1\bowtie K_1$.
\end{proof}

A \emph{cograph} is a graph that can be constructed from a given solitary vertex using the repeated application of the disjoint union and join operations. Another standard characterization of cographs is that they are those graphs that do not contain a four-vertex path as an induced subgraph. All complete graphs, complete bipartite graphs, cluster graphs and threshold graphs are cographs. Since any cograph $G$ has the property that either $G$ or $\overline{G}$ is disconnected, we can conclude that any cograph having an $OLD_{oind}$-set can be obtained from either $K_2$ or $K_3$ by the operations described in Theorem \ref{p4-tidy-2}. Hence, Corollary \ref{coro-cographs} follows.

\begin{mycorollary}\label{coro-cographs}
Let $G$ be a connected cograph of order $n\geq 2$. $G$ has an $OLD_{oind}$-set iff $G$ is either (i) isomorphic to either $K_2$ or $K_3$ or (ii) it can be obtained from them recursively by applying the following operation. Let $t\geq 2$ and $G_1, \ldots, G_t$ be connected cographs each having an $OLD_{oind}$-set. Set $G:=(G_1 \oplus \cdots \oplus G_t)\bowtie K_1$.
\end{mycorollary}

\section{Complementary prisms of cographs}
\label{sec:compl-prism of cographs}


Corollary \ref{coro-cographs} provides a characterization of the class of cographs that have an $OLD_{oind}$-set.
We now consider $OLD_{oind}$-sets in the complementary prisms of cographs.
If $G$ is a cograph then $\overline{G}$ and $G\oplus\overline{G}$ are also cographs. However, if $G$ is a nontrivial cograph, $G\overline{G}$ is is a $P_7$-free and is not a cograph. 

If $G$ is a connected cograph, then $\overline{G}$  is disconnected. Henceforth it is assumed that $G$ is a connected cograph and we denote the connected components of $\overline{G}$ by $\overline{G_1}, \overline{G_2}, \ldots, \overline{G_t}$; and the anticomponents of $G$ by $G_1, G_2, \ldots, G_t$. The number of vertices of the subgraph $G_i$ (resp. $\overline{G_i}$) is denoted by $|G_i|$ (resp. $|\overline{G_i}|$). 
There are infinite families of cographs that have complementary prisms with $OLD_{oind}$-sets, for example, the family described in Theorem \ref{OO-1-univ-vert} below.

\begin{mytheorem} \label{OO-1-univ-vert} \cite{cappelle2019}
If $G$ is a nontrivial graph with a unique universal vertex, then  $G\overline{G}$ has an $OLD_{oind}$-set iff $G=K_1\bowtie \overline{m K_2}$, where $m \geq 1$.
\end{mytheorem}

Cappelle et al. \cite{cappelle2019} reported some properties of an $OLD_{oind}$-set in a complementary prism $G\overline{G}$. For instance, they proved that in any $OLD_{oind}$-set in this class of graphs there is at most one edge that directly connects a vertex in $V(G)$ with a vertex in $V(\overline{G})$. The authors also proved that, given a general graph $G$, deciding whether or not $G\overline{G}$ has an $OLD_{oind}$-set is an $\mathcal{NP}$-complete problem \cite{cappelle2020}. However, the special case where $G$ has girth of at least four can be decided in polynomial time. 
Here, we show that the connected cographs having an $OLD_{oind}$-set $S$ such that $|S\cap V(G)|=1$ are exactly those described in Theorem \ref{OO-1-univ-vert}.

By Theorem \ref{OO-1-univ-vert}, if $G$ is a nontrivial cograph with a unique universal vertex, then  $G\overline{G}$ has an $OLD_{oind}$-set iff either (i) $|\overline{G_i}|=1$, for a unique $i$, where $1 \leq i \leq t$, or (ii) $|\overline{G_i}|=2$. If a cograph $G$ does not have a universal vertex, then $|\overline{G_i}|\geq 2$,  $ i \in [t]$. In this case, we show that if $G\overline{G}$ has an $OLD_{oind}$-set $S$, then at least one, and at most three, components of $\overline{G}$ have at least three vertices. Let $S=S_0\cup\overline{S}_1$ with $S_0\subseteq V(G)$ and $\overline{S}_1\subseteq V(\overline{G})$. We consider the case $|S_0|=2$ in Theorem \ref{cografo-oo-2} and the general case in Theorem \ref{cograph}.
For each component $\overline{G_i}$ of $\overline{G}$, let $\overline{D_i}=\overline{G_i}\setminus \overline{S}_1$ and $\overline{D}=\bigcup_{i\in[t]}\overline{D_i}$. Let $D_i$ (resp. $D$) be the corresponding vertices of $\overline{D_i}$ (resp. $\overline{D}$) in $G_i$ (resp. $G$), with $S_0 \subseteq D$. Let $n_i = |V(G_i)|$. 

\begin{restatable}{mytheorem}{cardSone}
\label{theo: card-S_0 = 1}
Let $G$ be a nontrivial connected cograph 
such that $G\overline{G}$ has an $OLD_{oind}$-set $S=S_0\cup\overline{S}_1$ with $S_0\subseteq V(G)$ and $\overline{S}_1\subseteq V(\overline{G})$. Then $|S_0|=1$ iff $G=K_1\bowtie \overline{m K_2}$, where $m \geq 1$.  
\end{restatable}
\begin{proof}
If $G=K_1\bowtie \overline{m K_2}$, where $m \geq 1$, it is easy to see that $G\overline{G}$ has an $OLD_{oind}$-set $S_0\cup\overline{S}_1$ say, where $\overline{S}_1=V(\overline{G})$ and $S_0$ contains only a universal vertex of $G$ and thus $|S_0|=1$. For the converse, suppose $G$ is such that $G\overline{G}$ has an $OLD_{oind}$-set $S_0\cup\overline{S}_1$ with $|S_0|=1$, where $S_0=\{v\}$, say. Since $G$ is nontrivial it has at least two vertices and every vertex in $V(G)\setminus\{v\}$ has at most one neighbor in $\overline{S}_1$. In order to be dominated at least twice, the neighbor must be in $N_G(v)$ and thus, $v$ is a universal vertex of $G$.  By Theorem \ref{OO-1-univ-vert}, $\overline{G}=K_1\oplus s K_2$, where $s \geq 1$.
\end{proof}

For each component $\overline{G_i}$ of $\overline{G}$, let $\overline{D_i}=\overline{G_i}\setminus \overline{S}_1$ and $\overline{D}=\bigcup_{i\in[t]}\overline{D_i}$. Let $D_i$ (resp. $D$) be the corresponding vertices of $\overline{D_i}$ (resp. $\overline{D}$) in $G_i$ (resp. $G$), with $S_0 \subseteq D$. Let $n_i = |V(G_i)|$. 

\begin{restatable}{mylemma}{lemmas}
\label{lemmas}
If $G$ is a connected cograph such that $|\overline{G_i}|\geq2$, $ i \in [t]$, where $t\geq2$, and that $G\overline{G}$ has an $OLD_{oind}$-set $S=S_0\cup\overline{S}_1$ with $S_0\subseteq V(G)$ and $\overline{S}_1\subseteq V(\overline{G})$, then the following statements hold:
\begin{enumerate} [label=(\roman*),topsep=0pt,itemsep=-1ex,partopsep=1ex,parsep=1ex]
    \item \label{theo:both-nonempty}$S_0 \neq \emptyset$ and $\overline{S}_1 \neq \emptyset$.
    \item \label{s1-nonempty} For each $i \in [t]$, $V(\overline{G_i})\cap \overline{S}_1 \neq \emptyset$.
    \item \label{cografo-Dbar-i} If $i\in[t]$, $n_i\geq3$ and $|S_0 \cap V(G_i)|\leq1$,  then $\overline{D_i}$ is an independent set of $\overline{G}$ with $1\leq|\overline{D_i}|\leq2$.
    \item \label{cografo-size-1} If $|S_0|=2$, then $G$ has two anticomponents $G_i$ and $G_j$ say, having nonempty intersection with $S_0$, and $n_i,n_j\geq3$. Moreover, $2\leq|\overline{D}|\leq 3$.
    \item \label{cografo-empty} There exists at most one $i\in[t]$ such that $n_i\geq3$ and $S_0 \cap V(G_i)=\emptyset$. 
    \item \label{cografo-three} $\overline{G}$ has at least one and at most three components of size at least three.  
\end{enumerate}
\end{restatable}
\begin{proof}
We prove each item separately.
\begin{enumerate} [label=\emph{(\roman*)}]
    \item For a proof by contradiction, suppose that $S_0=\emptyset$. Let $v \in V(G)$. By Theorem \ref{OO-general}, vertex $v$ has to be open-dominated at least twice. Since $v$ has at most one neighbor in $\overline{S}_1$, we can conclude that $S$ is not an $OLD_{oind}$-set in $G\overline{G}$. The proof for $\overline{S}_1$ follows analogously.

    \item For a proof by contradiction, suppose that there is an $i \in [t]$ such that $V(\overline{G_i})\cap \overline{S}_1 = \emptyset$. Let $\overline{v} \in V(\overline{G_i})$. By Theorem \ref{OO-general}, vertex $\overline{v}$ has to be open-dominated at least twice. Since $\overline{v}$ has at most one neighbor in ${S}_0$, we can conclude that $S$ is not an $OLD_{oind}$-set in $G\overline{G}$.

    \item If $\overline{D_i}$ is empty, as $G$ is connected, $\overline{G_i}$ is isomorphic to $K_2$ and thus $n_i=2$. Suppose $|\overline{D_i}|\geq3$. Then the vertices in $D_i$ have the same neighborhood as those in in $G\setminus V(G_i)$. Since $|V(G_i) \cap S_0|\leq1$, $D_i$ has at least two vertices that are not distinguished by $S$. Now, suppose that $\overline{D_i}$ has two vertices and they are adjacent in $\overline{G}$. In this case the vertices in $D_i$ are not adjacent in $G_i$. If one of them is a member of $S_0$, the other vertex is dominated once, since $S_0$ is open independent, it can possess at most one vertex of another anticomponent $G_j$ say, with $i\neq j$. If neither of the vertices is a member of $S_0$, they have the same neighborhood in $S_0$ and they are not distinguished by $S$.

    \item Suppose $|S_0|=2$. We first prove that $G$ has two anticomponents $G_i$ and $G_j$ say, that have nonempty intersection with $S_0$, and that $n_i,n_j\geq3$. For a proof by contradiction, suppose that there is an unique anticomponent $G_i$ of $G$ with $|G_i\cap S_0|=2$. Let $S_0=\{u,v\}$. Since $G_i$ is disconnected, there is a vertex $z \in V(G_i)$ that is adjacent to neither $u$ nor $v$. Since $z$ has at most one neighbor in $\overline{S}_1$, it follows that $S$ is not an $OLD_{oind}$-set. Now, suppose that $|V(G_i)\cap S_0|=1$ and that $n_i=2$. Let $V(G_i)=\{u,v\}$. Without loss of generality, assume $S_0=\{u\}$. If $\overline{u}\in \overline{S}_1$, then $\overline{v}\notin \overline{S}_1$ and $\overline{v}$ is dominated once by $S$, contradicting the fact that $S$ is an $OLD_{oind}$-set of $G\overline{G}$. If $\overline{u}\notin \overline{S}_1$, since it has to be dominated twice, $\overline{v}\in  \overline{S}_1$ and, since $S_0$ is open independent, $v\in S_0$. Thus $u,v \in S_0$, which contradicts the fact that there are two anticomponents in $G$ with nonempty intersection with $S_0$.

    Next, we prove that $2\leq|\overline{D}|\leq 3$. By \ref{theo:both-nonempty} above, there are two anticomponents of ${G}$ of size at least three, $G_i$ and $G_j$ say, that have nonempty intersection with $S_0$. Since $|\overline{D_i}|,|\overline{D_j}|\geq1$, it follows that $|\overline{D}|\geq 2$. For the upper bound, since $|S_0 \cap V(G_k)| \leq 1$, for every $k\in [t]$, by Lemma \ref{lemmas} \ref{cografo-Dbar-i} the vertices of $\overline{D}_k$ are independent in $\overline{G}$.  Also, as the vertices of the distinct components of $\overline{G}$ are not adjacent, $\overline{D}$ is an independent set of $\overline{G}$ and ${D}$ induces a complete graph in $G$. If $|D|\geq4$, since at most two of these vertices are in $S_0$, and they have no neighbors in $\overline{S}_1$, the remainder (at least two vertices) are not distinguished by $S$.

    \item Suppose that there are two anticomponents of $G$ of order at least 3, say $G_i$ and $G_j$, that do not have vertices in common with $S_0$. By \ref{cografo-Dbar-i}, $|{D}_i \cup {D}_j|\geq2$. However, the vertices in $D_i\cup D_j$ are not distinguished by $S_0$, since they have the same neighborhood in $G\setminus(V(G_i) \cup V(G_j))$.

    \item Firstly, suppose that each component of $\overline{G}$ is a $K_2$. By  \ref{theo:both-nonempty}, $S_0\neq \emptyset$. Suppose that $v \in V(G_i)$ with $v \in S_0$. If $\overline{v}\in \overline{S}_1$, its unique neighbor in $\overline{G}$, $\overline{u}$ say, does not belong to $\overline{S}_1$ and then, in order to dominate  $\overline{u}$ twice, $u \in S_0$. This implies that $u$ has another neighbor in $S_0$ that is a member of an anticomponent different from $G_i$. This contradicts the assumption that $S$ is open-independent. If $\overline{v} \notin \overline{S}_1$, the vertex $\overline{u}$ is in $\overline{S}_1$ and we have an analogous situation. Thus, it can be concluded that $S_0=\emptyset$,  which contradicts the premise of \ref{theo:both-nonempty}. Therefore, $\overline{G}$ has at least one component of size at least three.
    Secondly, for a proof by contradiction, suppose that at least four components of $\overline{G}$ have size at least three.
    Since the elements of $S_0$ are open-independent, $S_0$ contains vertices that are members of at most two distinct anticomponents of $G$. Thus, there are at least two anticomponents of $G$ with order at least 3, say $G_i$ and $G_j$, that do not have vertices in common with $S_0$. By \ref{cografo-empty}, $S$ is not an $OLD_{oind}$-set.
\end{enumerate}
\end{proof}

For $\ell,m\geq 1$, the graph  $\overline{K_{\ell}}\bowtie mK_2$ that contains $\ell+2m$ vertices is denoted by $R_{\ell,m}$. Also,  $R^*_{\ell,m}$ denotes the graph obtained from $R_{\ell,m}$ with one edge missing between one vertex of $\overline{K_{\ell}}$ and exactly one copy of $K_2$. Note that $R_{\ell,m}$ is a cograph; $R^*_{\ell,1}$ for $\ell=1,2$, is also a cograph; $R_{1,1}$ is isomorphic to $K_3$; and $R^*_{1,1}$ is isomorphic to $P_3$. 

\begin{restatable}{mylemma}{lemmanigeqthree}
\label{lemma-ni>3}
Let $G$ be a connected cograph such that $|\overline{G_i}|\geq2$, $ i \in [t]$, where $t\geq2$, and that $G\overline{G}$ has an $OLD_{oind}$-set $S=S_0\cup\overline{S}_1$ with $S_0\subseteq V(G)$ and $\overline{S}_1\subseteq V(\overline{G})$ and let $i \in [t]$ such that $n_i\geq3$. 
\begin{enumerate} [label=(\roman*),topsep=0pt,itemsep=-1ex,partopsep=1ex,parsep=1ex]
    \item \label{cograph-R1m} If $V(G_i) \cap S_0 = \emptyset$, then $\overline{G_i} \cong R_{1,m}$. 
    \item \label{cografo-so=1} If $|V(G_i) \cap S_0| = 1$, then $\overline{G_i}$ is isomorphic to one of $R^*_{\ell,1}$ and $R_{\ell,m}$, for $\ell=1,2$.  
\end{enumerate}
\end{restatable}
\begin{proof}

To prove \ref{cograph-R1m}, suppose that $n_i\geq3$ and {$V(G_i \cap S_0 = \emptyset)$}. By Lemma \ref{lemmas} \ref{s1-nonempty}, $V(\overline{G_i}) \cap \overline{S}_1 \neq \emptyset$. Let $\overline{A_i}= V(\overline{G_i}) \cap \overline{S}_1$. By Lemma \ref{lemmas} \ref{cografo-Dbar-i}, $\overline{D_i}$ induces an independent set with $1\leq|\overline{D_i}| \leq2$. The set $\overline{A_i}$ induces  $m=|\overline{A_i}|/2$ independent edges.  Since $\overline{G_i}$ is connected, at least one vertex of every edge induced by $\overline{A_i}$ has a neighbor in $\overline{D_i}$.
If $|\overline{D_i}|=2$, the two vertices of $D_i$ are not distinguished by $S$, since they have the same neighborhood in $G\setminus V(G_i)$. Thus, we may assume $|\overline{D_i}|=1$. Let $\overline{D_i}=\{\overline{u}\}$. 
If $m=1$, since $\overline{u}$ has at least two neighbors in $\overline{S}_1$, it follows that $\overline{G}\cong R_{1,1}$. If $m\geq2$, we claim that $\overline{u}$ is adjacent to every vertex in $\overline{A_i}$. Suppose that this is not true and consider an edge $ab$ induced by vertices in $\overline{A_i}$ such that $\overline{u}a \notin E(\overline{G_i})$. In this case, $ab\overline{u}c$ is an induced 4-path in $\overline{G_i}$, where $c$ is an vertex of some other edge induced by $\overline{A_i}$. So, $\overline{G_i}$ is not a cograph.
Therefore, $\overline{G_i} \cong R_{1,m}$.

Now, we prove \ref{cografo-so=1}. By Lemma \ref{lemmas} \ref{s1-nonempty}, $V(\overline{G_i}) \cap \overline{S}_1 \neq \emptyset$.  Let $\overline{A_i}= V(\overline{G_i}) \cap \overline{S}_1$. By Lemma \ref{lemmas} \ref{cografo-Dbar-i}, $\overline{D_i}$ induces an independent set with $1\leq|\overline{D_i}| \leq2$. The set $\overline{A_i}$ induces  $m=|\overline{A_i}|/2$ independent edges.  Since $\overline{G_i}$ is connected, at least one vertex incident with every edge induced by $\overline{A_i}$ has a neighbor in $\overline{D_i}$. 
If $m=1$, we have four possible cographs, $R^*_{\ell,1}$, $R_{\ell,1}$, for $\ell=1,2$. If $m\geq2$, we claim that every vertex in  $\overline{D_i}$ is adjacent to every vertex in $\overline{A_i}$. For a proof by contradiction, suppose that this is not true and consider two edges $ab$ and $cd$, induced by vertices in $\overline{A_i}$. Let $\overline{u} \in \overline{D_i}$. Without loss of generality, suppose that $\overline{u}a \notin E(\overline{G_i})$, for an edge $ab$ induced by vertices in $\overline{A_i}$. In this case, $ab\overline{u}c$ is an induced 4-path in $\overline{G_i}$, where $c$ is an vertex of some other edge induced by $\overline{A_i}$. So, $\overline{G_i}$ is not a cograph. Therefore,  $\overline{G_i}$ is isomorphic to $R_{\ell,m}$, for $\ell=1,2$.
\end{proof}

Next, we present in Theorems \ref{cografo-oo-2} and \ref{cograph} a general recursive characterization of the class $C$ of connected cographs such that if $G$ is a member of $C$ then $G\overline{G}$ has an $OLD_{oind}$-set. Theorem \ref{cografo-oo-2} (\ref{cograph}) is illustrated in Figure \ref{fig:cp-cograph2} (\ref{fig:cp-cograph1}).

\begin{restatable}{mytheorem}{cografoOOtwo}
\label{cografo-oo-2}
Let $G$ be a connected cograph such that $|\overline{G_i}|\geq2$, $ i \in [t]$ with $t\geq2$. $G\overline{G}$ has an $OLD_{oind}$-set $S=S_0 \cup \overline{S}_1$ with $S_0\subseteq V(G)$ and $\overline{S}_1\subseteq V(\overline{G})$  such that $|S_0|=2$ iff there are at least two, and at most three, components $\overline{G_i}$ with $n_i\geq3$ and each component of $\overline{G}$ is isomorphic to one of the graphs $K_2$, $R^*_{\ell,1}$ or $R_{\ell,m}$, for $\ell=1,2$ and $m\geq 1$, with the following further conditions:
\begin{enumerate} [label=(\roman*),topsep=0pt,itemsep=-1ex,partopsep=1ex,parsep=1ex]
    \item \label{cografo-oo-2(i)} at most one component of $\overline{G}$ is either $R^*_{2,1}$ or $R_{2,m}$ and if $\overline{G}$ has such a component, $\overline{G}$ has exactly two components of size at least three; and
    \item \label{cografo-oo-2(ii)} at most two components of $\overline{G}$ are isomorphic to $R^*_{1,1}$, $R^*_{2,1}$ or $R_{2,m}$.  
\end{enumerate}
\end{restatable}
\begin{proof} 
Suppose that $G$ is as described and that $G\overline{G}$ has an $OLD_{oind}$-set $S=S_0 \cup \overline{S}_1$ with $S_0\subseteq V(G)$, $\overline{S}_1\subseteq V(\overline{G})$, and $|S_0|=2$. Let $t'$ be the number of components of $\overline{G}$ of size at least three. By Lemma \ref{lemmas} \ref{cografo-three}, $t' \leq 3$. By Lemma \ref{lemmas}  \ref{cografo-size-1}, two anticomponents of $G$ have nonempty intersection with $S_0$ and both have size at least three, thus $2\leq t' \leq 3$.  

Since $|V(G_i) \cap S_0|\leq 1$, for every component $\overline{G_i}$ of $\overline{G}$ with $n_i\geq 3$, by Lemma \ref{lemma-ni>3} \ref{cografo-so=1}, the components of $\overline{G}$ of size at least three are $R^*_{\ell,1}$ or $R_{\ell,m}$, for $\ell=1,2$. To prove \ref{cografo-oo-2(i)}, firstly suppose that there are at least two components of $\overline{G}$ that are one of the graphs $R^*_{2,1}$ and $R_{2,m}$, say $\overline{G_i}$ and $\overline{G_j}$. By Lemma \ref{lemmas} \ref{cografo-Dbar-i}, $\overline{D_i}$ (resp. $\overline{D_j}$) is independent with size at least one and at most two. If $|\overline{D_i}|=1$ (resp. $|\overline{D_j}|=1$) the remaining vertices of $\overline{G_i}$ (resp. $\overline{G_j}$) cannot induce independent edges. Hence $|\overline{D_i}|=|\overline{D_j}|=2$ which implies that $|\overline{D}|\geq 4$. By Lemma \ref{lemmas}  \ref{cografo-size-1},  $|S_0| \neq 2$, contradicting the premise of the theorem. Hence we can conclude that at most one component of $\overline{G}$ is either $R^*_{2,1}$ or $R_{2,m}$. Secondly, suppose that $\overline{G}$ has exactly one component, say $\overline{G_i}$, which is isomorphic to either $R^*_{2,1}$ or $R_{2,m}$. Thus $|\overline{D_i}|=2$. Suppose that $\overline{G}$ has other two components, say $\overline{G_j}$ and $\overline{G_p}$, each one with size at least three. By Lemma \ref{lemmas} \ref{cografo-Dbar-i}, $|\overline{D_j} \cup \overline{D_p}|\geq2$. Since $|\overline{D_i}|=2$, $|\overline{D}|\geq 4$. Again, by Lemma \ref{lemmas} \ref{cografo-size-1}, $|S_0| \neq 2$, contradicting the premise of the theorem. This completes the proof of \ref{cografo-oo-2(i)}.
Now we prove \ref{cografo-oo-2(ii)}. If $t'=3$, since $|S_0|=2$ one of the components of $\overline{G}$ with size at least three, say $\overline{G_j}$, is such that $G_j$ has empty intersection with $S_0$. Thus by Lemma \ref{lemma-ni>3} \ref{cograph-R1m}, $\overline{G_j} \cong R_{1,m}$, which implies that at most two components of $\overline{G}$ are isomorphic to the graphs $R^*_{1,1}$, $R^*_{2,1}$, and $R_{2,m}$.
The difference $t-t'$ is the number of components of $\overline{G}$ of size two. Clearly, each of these components (if any) is isomorphic to $K_2$.

Conversely, suppose $G\overline{G}$ is as described, with at least two, and at most three, components $\overline{G_i}$ with $n_i\geq3$, where the components of $\overline{G}$ are the graphs $K_2$, $R^*_{\ell,1}$, and $R_{\ell,m}$, for $\ell=1,2$. Assume that conditions \ref{cografo-oo-2(i)} and \ref{cografo-oo-2(ii)} hold. We show how to select the two vertices in $S_0$ and the vertices in $\overline{S}_1$ such that $S = S_0 \cup \overline{S}_1$ is an $OLD_{oind}$-set of $G\overline{G}$. For every component $\overline{G_i}$ of $\overline{G}$ there is a partition of this vertex set into sets $\overline{D_i}$ and $\overline{M_i}$ such that $\overline{D_i}$ is an independent set with $|\overline{D_i}|\in \{0,1,2\}$ and $\overline{M_i}$ contains the $2m$  vertices of the $m\geq1$ independent edges of $\overline{G_i}$. Let $\overline{D_i}$ and $\overline{M_i}$ be such sets.
Let $\overline{D}=\bigcup_{i\in[t]}\overline{D_i}$.
Then, $\overline{S}_1=\bigcup_{i\in [t]}M_i=\overline{M}$. 
Let $D_i$, $M_i$, $D$, and $M$ be the sets of corresponding vertices in $G$.
Note that in $G$, every vertex in $M$ has a different neighbor in $\overline{M}$. Since we shall show that $S_0$ has vertices in two distinct anticomponents, every vertex of $G$ has at least one neighbor in $S_0$. Hence, the vertices in $M \cup \overline{M}$ are distinguished and open-dominated by $S$ and we need to consider only the vertices in $D\cup\overline{D}$.

We proceed by setting the vertices in $S_0$.
We first consider the case $t'=2$. Without loss of generality, assume $|G_1|,|G_2| \geq 3$. Since conditions \ref{cografo-oo-2(i)} and \ref{cografo-oo-2(ii)} hold, by symmetry, we can assume that $\overline{G_1} \in \{R^*_{1,1}, R^*_{2,1}, R_{1,m},R_{2,m}\}$ and $\overline{G_2} \in \{R^*_{1,1}, R_{1,m}\}$. 

If $|D_1|=1$, let $D_1=\{v\}$. If $|D_1|=2$ and $\overline{G_1}$ is isomorphic to $R^*_{2,1}$, let $D_1=\{u,v\}$, with $v$ being the vertex of degree two in $G_1$. If $|D_1|=2$ and $\overline{G_1}$ is isomorphic to $R_{2,m}$, let $D_1=\{u,v\}$. Let $\overline{u},\overline{v}$ be the corresponding vertices of $u,v$ in $\overline{G_1}$. We set $v\in S_0$, which gives rise to the following cases:
\begin{enumerate} [topsep=0pt,itemsep=-1ex,partopsep=1ex,parsep=1ex]
    \item \label{case-a} If $\overline{G_1}$ is isomorphic to $R^*_{1,1}$, then $\overline{v}$ has one neighbor in $\overline{M_1}$ and $v \in N_G(\overline{v})$.
    \item If $\overline{G_1}$ is isomorphic to $R^*_{2,1}$, then $\overline{u}$ has two neighbors in $\overline{M_1}$ and  $\overline{v}$ has exactly one neighbor in $\overline{M_1}$. Since $v \in S_0$, $\overline{u}$ and $\overline{v}$ are open-dominated twice and distinguished by $S$.
    \item If $\overline{G_1}$ is isomorphic to $R_{2,m}$, then $\overline{u}$ and $\overline{v}$ are both neighbors of every vertex in $\overline{M_1}$. Since $v \in S_0$, these two vertices are open-dominated and distinguished by $S$.
    \item \label{case-d} If $\overline{G_1}$ is isomorphic to $R_{1,m}$, then $\overline{v}$ is the unique vertex that is the neighbor of all vertices of $M_1$.
\end{enumerate}
As $\overline{G_2} \in \{R^*_{1,1}, R_{1,m}\}$, $|D_2|=1$. Let $D_2=\{v\}$. We set $v \in S_0$. We have two possibilities for $\overline{G_2}$ that are analogous to the cases \ref{case-a} and \ref{case-d}, described above. 
We can conclude that all vertices in $\overline{G}$ are distinguished and open-dominated by $S$, and $S$ is open-independent.

Now we consider the vertices in $D$. Since at most one component of $\overline{G}$ is isomorphic to either $R^*_{2,1}$ or $R_{2,m}$, it follows that $2\leq|D|\leq 3$. Two of the vertices in $D$ are also in $S_0$. So, there is at most one vertex of $G$, $z$ say, such that neither $z$ nor $\overline{z}$ is in $S$. The choosing of $S_0$ guarantees that $z$ is the only vertex of $G$ that is dominated only by $S_0$.

Finally, we consider the case $t'=3$. Without loss of generality, we assume that $|G_j| \geq 3$, for every $j\in[3]$. Since conditions \ref{cografo-oo-2(i)} and \ref{cografo-oo-2(ii)} hold, and by symmetry, we assume that $\overline{G_1},\overline{G_2} \in \{R^*_{1,1}, R_{1,m}\}$ and $\overline{G_3}$ is isomorphic to $R_{1,m}$. Thus, $|D_1|=|D_2|=1$ and we set $S_0=D_1 \cup D_2$. 

As $\overline{G_1},\overline{G_2} \in \{R^*_{1,1}, R_{1,m}\}$, we can use similar arguments to those used in the case $t'=2$ in order to conclude that each vertex in $\overline{D_1}\cup\overline{D_2}$ is both dominated at least twice and distinguished by $S$.
Since $S_0=D_1\cup D_2$, these vertices are distinguished by $S$. Again, the vertex in $D_3$ is the only vertex of $G$ that is dominated only by $S_0$.
Therefore, $S$ is an $OLD_{oind}$-set of $G\overline{G}$ with $|S_0|=2$. 
See Fig. \ref{fig:cp-cograph2} for an example.
\end{proof}

\begin{figure}[!ht]
\centering
 \begin{subfigure}[b]{.9\textwidth}
 \centering
  \begin{tikzpicture}[
 graph/.style={matrix of math nodes, ampersand replacement=\&, column sep=12pt, row sep=15pt,  nodes={circle,thin,inner sep=0pt,minimum size=4pt}},
 oldSet/.style args = {(#1,#2)}{row #1 column #2/.style={nodes={draw,fill=black}}},
 nodeSet/.style args = {(#1,#2)}{row #1 column #2/.style={nodes={draw,fill=black!20!white}}},
 myBlock/.style = {black,thin,draw=#1, inner sep=4pt, rounded corners},
 every label/.append style={font=\scriptsize,label distance=-1pt},
 oldSet/.list={(1,1),(1,2),(1,5),(1,6),(1,8),(1,9),(1,10),(1,11),(2,4),(2,7)},
 nodeSet/.list={(1,3),(1,4),(1,7),(2,1),(2,2),(2,3),(2,5),(2,6),(2,8),(2,9),(2,10),(2,11)}
]
\path[use as bounding box] (-3.6,-1.1) rectangle (4.4, 0.9);
\matrix [graph] (M)
{
 {}\& {}\& {}\& {}\& {}\& {}\& [10pt] {}\& {}\& {}\& [10pt]{}\& {}\\
 {}\& {}\& {}\& {}\& {}\& {}\&  {}\& {}\& {}\&{}\& {}\\[-7pt]
 {}\& {}\& {}\& {}\& {}\& {}\&  {}\& {}\& {}\&{}\& {}\\
};
\foreach \i  in {1,...,11}{
 \draw [thin] (M-1-\i) -- (M-2-\i) {};
}
\foreach \i [evaluate=\i as \j using int(\i+1)] in {1,2,4,5,7,8,10}
 \draw [thin] (M-1-\i) -- (M-1-\j);
\draw [thin] (M-2-3) -- (M-2-4);
\node also [label={[xshift=7pt]right:$V(\overline{G})$}] (M-1-11) {};
\node[myBlock,fit=(M-2-1) (M-3-6),label={below:$G_1$}] (b21) {};
\node[myBlock,fit=(M-2-7) (M-2-9),label={below:$G_2$}] (b22) {};
\node[myBlock,fit=(M-2-10) (M-2-11),label={below:$G_3$},label={[xshift=2pt]right:$V(G)$}] (b23) {};
\draw [thin] 
(M-1-1) to[out=40,in=140] (M-1-3)
(M-1-1) to[out=40,in=140] (M-1-4)
(M-1-2) to[out=40,in=140] (M-1-4)
(M-1-3) to[out=30,in=140] (M-1-5)
(M-1-3) to[out=40,in=140] (M-1-6)
(M-1-4) to[out=40,in=150] (M-1-6)
(M-2-1) to[out=330,in=210] (M-2-5)
(M-2-1) to[out=320,in=220] (M-2-6)
(M-2-2) to[out=340,in=200] (M-2-5)
(M-2-2) to[out=330,in=210] (M-2-6)
(M-2-7) to[out=335,in=205] (M-2-9);
\draw [thick]
(b21) to[out=330,in=230,looseness=.5] (b23)
(b21) to[out=6.1,in=183,looseness=0] (b22)
(b22) to (b23);
\end{tikzpicture}
    \caption{The components of $\overline{G}$ are the graphs $R_{2,2}$, $R^*_{1,1}$, and $K_2$.}
    \label{fig:cp-cograph2}
 \end{subfigure}
 
 \begin{subfigure}[b]{.9\textwidth}
 \centering
    \begin{tikzpicture}[
 graph/.style={matrix of math nodes, ampersand replacement=\&, column sep=12pt, row sep=25pt,  nodes={circle,thin,inner sep=0pt,minimum size=4pt}},
 oldSet/.style args = {(#1,#2)}{row #1 column #2/.style={nodes={draw,fill=black}}},
 nodeSet/.style args = {(#1,#2)}{row #1 column #2/.style={nodes={draw,fill=black!20!white}}},
 myBlock/.style = {black,thin,draw=#1, inner sep=4pt, rounded corners},
 every label/.append style={font=\scriptsize,label distance=-1pt},
 oldSet/.list={(1,1),(2,1),(1,2),(1,3),(1,4),(1,5),(1,8),(1,9),(1,14),(1,15),(1,16),(1,17),(2,6),(2,7),(2,10),(2,11)},
 nodeSet/.list={(1,6),(1,7),(1,10),(1,11),(1,12),(1,13),(2,2),(2,3),(2,4),(2,5),(2,8),(2,9),(2,12),(2,13),(2,14),(2,15),(2,16),(2,17)}
]
\path[use as bounding box] (-5.90,-1.4) rectangle (6.6, 1.2);
\matrix [graph] (M)
{
 {}\& {}\& {}\& {}\& {}\& [8pt] {}\& {}\& [8pt]{}\& {}\& [8pt]{}\& {}\& {}\&[12pt] {}\& {}\& {}\& [12pt]{}\& {}\\
 {}\& {}\& {}\& {}\& {}\& {}\&[8pt] {}\& {}\& {}\&[12pt] {}\& {}\& {}\&[8pt] {}\& {}\& {}\& {}\& {}\\[-20pt]
 {}\& {}\& {}\& {}\& {}\& {}\&[8pt] {}\& {}\& {}\&[12pt] {}\& {}\& {}\&[8pt] {}\& {}\& {}\& {}\& {}\\
};
\foreach \i in {1,...,17}
 \draw [thin] (M-1-\i) -- (M-2-\i) {};
\foreach \i [evaluate=\i as \j using int(\i+1)] in {2,4,8,13,14,16}
 \draw [thin] (M-1-\i) -- (M-1-\j);
\foreach \i [evaluate=\i as \j using int(\i+1)] in {1,3,6,10,11}
 \draw [thin] (M-2-\i) -- (M-2-\j);
\node also [label={[xshift=7pt]right:$V(\overline{G})$}] (M-1-17) {};
\node[myBlock,fit=(M-1-1)  (M-1-5)] (b11) {};
\node[myBlock,fit=(M-1-6)  (M-1-7)] (b12) {};
\node[myBlock=black,inner sep=3pt,fit=(b11) (M-1-9)] (b13) {};
\node[myBlock,fit=(M-1-10) (M-1-12)] (b14) {};
\node[myBlock,fit=(M-2-1) (M-3-7)] (b21) {};
\node[myBlock,fit=(M-2-8) (M-2-9)] (b22) {};
\node[myBlock,inner sep=3pt,fit=(b21) (M-2-12),label={below:$G_1$}] (b23) {};
\node[myBlock=black,fit=(M-2-13) (M-2-15),label={below:$G_2$}] (b24) {};
\node[myBlock=black,fit=(M-2-16) (M-2-17),label={below:$G_3$},label={[xshift=2pt]right:$V(G)$}] (b25) {};
\draw [thin] 
 (M-1-13) to[out=040,in=140] (M-1-15)
 (M-2-1)  to[out=340,in=220] (M-2-3)
 (M-2-1)  to[out=330,in=220] (M-2-4)
 (M-2-1)  to[out=320,in=220] (M-2-5)
 (M-2-2)  to[out=330,in=210] (M-2-4)
 (M-2-2)  to[out=320,in=210] (M-2-5)
 (M-2-3)  to[out=320,in=200] (M-2-5)
 (M-2-10) to[out=320,in=220] (M-2-12);
\draw [thick]
(b11) -- (b12)
(b21) to[out=4.4,in=181,looseness=0] (b22)
(b13) -- (b14)
(b23) to[out=2.9,in=182,looseness=0] (b24)
(b24) -- (b25)
(b23) to[out=340,in=210,looseness=.7] (b25);
\end{tikzpicture}
    \caption{$G$ can be obtained recursively from ${K_1\bowtie \overline{2K_2}}$. $G=\overline{\overline{\overline{\overline{K_1\bowtie \overline{2K_2}}\bowtie\overline{K_2}}\bowtie\overline{K_2}}\bowtie\overline{R_{1,1}}}\bowtie\overline{R_{1,1}}\bowtie\overline{K_2}=((((K_1\bowtie \overline{2K_2})\oplus  K_2)\bowtie\overline{K_2})\oplus{R_{1,1}})\bowtie\overline{R_{1,1}}\bowtie\overline{K_2}$.}
    \label{fig:cp-cograph1}
 \end{subfigure}
\caption{The complementary prisms of cographs having $OLD_{oind}$-sets. \ref{fig:cp-cograph2} illustrates Theorem \ref{cografo-oo-2} and \ref{fig:cp-cograph1} illustrates Theorem \ref{cograph}. An edge between two rectangles indicates adjacency between all pairs of vertices, one in each rectangle.}
\label{fig:complprisms}
\end{figure}
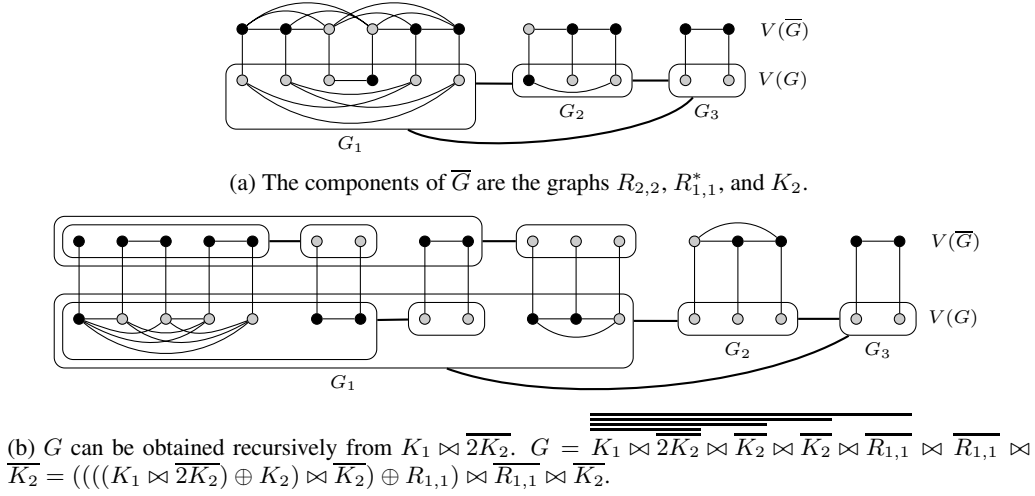

\begin{mytheorem} \label{cograph}
If $G$ is a connected cograph then $G\overline{G}$ has an $OLD_{oind}$-set iff either (i) $G$ is isomorphic to one of the graphs ${K_1\bowtie \overline{mK_2}}$, $m\geq 1$, and graph ${G}$ is as described in Theorem \ref{cografo-oo-2}, or (ii) $G$ can be obtained from them recursively by the following operation, where $H$ is a connected cograph such that $H\overline{H}$ has an $OLD_{oind}$-set. Set $G=\overline{H}\bowtie F \bowtie (\overline{rK_2})$, \noindent where $F \in \{\overline{R_{1,m}}, \overline{K_2}\}$, $r\geq0$, $m\geq1$.
\end{mytheorem}
\begin{proof}
Suppose that $G$ is a connected cograph such that $G\overline{G}$ has an $OLD_{oind}$-set $S=S_0 \cup \overline{S}_1$ where $S_0\subseteq V(G)$ and $\overline{S}_1\subseteq V(\overline{G})$. Assume that $\overline{G}$ has components $\overline{G_1}\ldots,\overline{G}_t$, with $t\geq2$. By Lemma \ref{lemmas} \ref{theo:both-nonempty}, $S_0$ is nonempty. If $|S_0|=1$, then by Theorem \ref{theo: card-S_0 = 1}, ${G}$ is the graph ${K_1\bowtie \overline{mK_2}}$, and if $|S_0|=2$, then $G$ is the graph described in Theorem \ref{cografo-oo-2}. So, we may assume that $|S_0|\geq 3$.
By Lemma \ref{lemmas} \ref{cografo-three}, at least one, and at most three, components of $\overline{G}$ have at least three vertices. Since $|S_0|\geq 3$, there exists exactly one anticomponent $G_i$ say, of $G$ with nonempty intersection with $S_0$. (If this were not so, since every vertex of $G_i$ is adjacent to every vertex of any other anticomponent $G_p$ with $p\neq i$, then $S$ would not be an open-independent set.) By Lemma \ref{lemmas} \ref{cografo-empty}, there exists at most one anticomponent $G_i$ say, of $G$ with size at least three with empty intersection with $S_0$. This implies that $G$ has at most two anticomponents $G_i$, with size at least three. If $G$ has an anticomponent $G_i$ of size at least three having empty intersection with $S_0$, by Lemma  \ref{lemma-ni>3} \ref{cograph-R1m}, this anticomponent is isomorphic to $\overline{R_{1,m}}$. Otherwise, if $|G_i| = 2$, then $G_i$ is isomorphic to $\overline{K_2}$.  
Without loss of generality, we assume that $G_1$ is an anticomponent of $G$ with $S_0 \subseteq V(G_1)$ and that $G_2$ is an anticomponent that is isomorphic to either $\overline{R_{1,m}}$ or to $\overline{K_2}$. Note that each of the $t-2$ anticomponents of $G$ (if any) is isomorphic to $\overline{K_2}$. So, until now, we have that $G=G_1\bowtie G_2 \bowtie (\overline{rK_2})$, with $r\geq0$, for some disconnected graph $G_1$ (since $\overline{G_1}$ is connected).
Now, suppose that $G_1\overline{G_1}$ does not have an $OLD_{oind}$-set.
Thus, $S'=S \cap V(G_1\overline{G_1})$ is not an $OLD_{oind}$-set of $G_1\overline{G_1}$, and there is a vertex $v$ say, of $G_1\overline{G_1}$ that is neither distinguished nor open-dominated by $S'$. Since $S_0 \subseteq V(G_1)$, and there is no edge between $\overline{G_1}$ and any other component of $\overline{G}$, we have that $N_{G\overline{G}}(v)\cap S\subseteq V(G_1\cup\overline{G_1})$. This implies that $v$ is neither distinguished nor dominated by $S$ and hence, $S$ is not an $OLD_{oind}$-set of $G\overline{G}$. Therefore, $G$ can be obtained by the described operation from a disconnected graph $G_1$ such that $G_1\overline{G_1}$ has an $OLD_{oind}$-set.

Conversely, suppose that $G$ is a connected cograph obtained as described in the statement of the theorem. Since, by Theorem \ref{theo: card-S_0 = 1}, ${K_1\bowtie \overline{mK_2}}$ has an $OLD_{oind}$-set, and is also the graph described in Theorem \ref{cografo-oo-2}, we assume that $G$ is none of these graphs. Hence we assume that $G=\overline{H}\bowtie F \bowtie (\overline{rK_2})$ and $H$ is a connected cograph such that $H\overline{H}$ has an $OLD_{oind}$-set $B=B_0 \cup \overline{B}_1$ with $B_0\subseteq V(H)$ and $\overline{B}_1\subseteq V(\overline{H})$. We show how to select the vertices of a set $S=S_0 \cup \overline{S}_1$ where $S_0\subseteq V(G)$ and $\overline{S}_1\subseteq V(\overline{G})$ such that $S$ is an $OLD_{oind}$-set of $G\overline{G}$. We begin setting $S_0=\overline{B}_1$ and $\overline{S}_1=B_0$. Hence, the vertices of $H\cup \overline{H}$ are all distinguished and open-dominated by $S$ and we need to consider only the remainder of the vertices. 
We analyze the two relevant cases.  In the first case, assume that $F$ is isomorphic to $\overline{R_{1,m}}$. Here, we must add to $\overline{S}_1$ the vertices that induce the $m$ independent edges in $\overline{F}=R_{1,m}$. Let $\overline{v}$ be the universal vertex in $\overline{F}$. Then $\overline{v}$ is dominated at least twice by the vertices incident with the $m\geq1$ edges of $\overline{F}$ and is the unique vertex with this neighborhood in $S$. The vertices in $F$ are adjacent to all the vertices in $S_0$. Since $\overline{H}$ is disconnected and hence, $H\overline{H}$ has an $OLD_{oind}$-set, by Lemma \ref{lemmas} \ref{s1-nonempty}, $\overline{B}_1$ has at least one vertex in each component of $\overline{H}$. Let $v$ be the neighbor of $\overline{v}$ in $F$. Each vertex in $V(F)\setminus\{v\}$ has a distinct neighbor in $\overline{S}_1$ and $v$ is the unique vertex whose neighborhood in $S$ is exactly $S_0$. Observe that $\overline{B}_1$ has size at least three and recall that $S_0$  was set to it. In the second case, assume that $F$ is isomorphic to $\overline{K_{2}}$. In this case we add the vertices of $\overline{F}=K_2$ to the set $\overline{S}_1$ and, by using similar arguments to those in the first case, it is easy to see that the vertices of $F$ are open-dominated and distinguished by $S$. If $t\geq 3$, each of the $t-2$ other components of $\overline{G}$ (if any) is isomorphic to $K_2$ and we proceed in the same way as in the former case. Therefore, $S$ is an $OLD_{oind}$-set of $G\overline{G}$. See Fig. \ref{fig:cp-cograph1} for an example of a graph $G\overline{G}$ obtained by this operation from ${K_1\bowtie \overline{2K_2}}$.
\end{proof}

\section{Summary}
\label{sec:summary}
The problem of deciding whether or not a graph $G$ has an $OLD_{oind}$-set has important applications and was shown above to be $\mathcal{NP}$-complete even for the special cases when $G$ is either a planar bipartite graph of maximum degree five and girth six, or a planar subcubic graph of girth nine. Characterizations of both the $P_4$-tidy graphs and the cographs that have $OLD_{oind}$-sets have been presented. Also, necessary and sufficient conditions for a complementary prism of a connected cograph to have an  $OLD_{oind}$-set are derived. For future work, it might be fruitful to study the complexity of identifying $OLD_{oind}$-sets in other families of graphs, such as those that are Hamiltonian, Eulerian or $n$-partite when $n \geq 3$.

\acknowledgements
The authors are grateful to the anonymous reviewers for their very helpful suggestions.

\nocite{*}
\bibliographystyle{abbrvnat}
\bibliography{old-oind}

\end{document}